%% file: distributedRDFDesign-edbt.tex
\documentclass{sig-alternate}
\usepackage{amsmath}
\usepackage{txfonts}
\usepackage{amssymb}
\usepackage{times}
\usepackage{graphicx}
\usepackage{epsfig,tabularx,amssymb,amsmath,subfigure,multirow}
\usepackage[linesnumbered,ruled,noend]{algorithm2e}
\usepackage[noend]{algorithmic}
\usepackage{multirow}
\usepackage{graphicx,floatrow}
\usepackage{listings}
\usepackage{threeparttable}
\usepackage{tikz}
\usepackage[T1]{fontenc}
\usepackage{pgfplots}
\newcommand{\nop}[1]{}
\newcommand{\tabincell}[2]{\begin{tabular}{@{}#1@{}}#2\end{tabular}}
\newtheorem{definition}{Definition}
\newtheorem{theorem}{Theorem}

\newtheorem{example}{Example}

\title{Query Workload-based RDF Graph Fragmentation and Allocation}

 \author{
{{Peng Peng${^1}$}, Lei Zou{${^{1,3}}\thanks{corresponding author: zoulei@pku.edu.cn}$}, Lei Chen{${^2}$}, Dongyan Zhao{${^{1,3}}$}}%
\vspace{1.6mm}\\
\fontsize{10}{10}\selectfont\itshape $~^{1}$Peking University,
China;\\ \fontsize{10}{10}\selectfont\itshape $~^{2}$ Hong Kong
University of Science and Technology, China; \\
\fontsize{10}{10}\selectfont\itshape $~^{3}$ Key Laboratory of Computational Linguistics (PKU), Ministry of Education, China \\
\fontsize{9}{9}\selectfont\ttfamily\upshape $\{$
pku09pp,zoulei,zhaodongyan$\}$@pku.edu.cn, leichen@cse.ust.hk
\\}

\begin{document}
\maketitle
\begin{abstract}
As the volume of the RDF data becomes increasingly large, it is essential for us to design a distributed database system to manage it. For distributed RDF data design, it is quite common to partition the RDF data into some parts, called \emph{fragments}, which are then distributed. Thus, the distribution design consists of two steps: fragmentation and allocation. In this paper, we propose a method to explore the intrinsic similarities among the structures of queries in a workload for fragmentation and allocation, which aims to reduce the communication cost during SPARQL query processing. Specifically, we mine and select some \emph{frequent access patterns} to reflect the characteristics of the workload. Based on the selected frequent access patterns, we propose two fragmentation strategies, vertical and horizontal fragmentation strategies, to divide RDF graphs while meeting different kinds of query processing objectives. Vertical fragmentation is for better throughput and horizontal fragmentation is for better performance. After fragmentation, we discuss how to allocate these fragments to various sites. Finally, we discuss how to process a query based on the results of fragmentation and allocation. Extensive experiments confirm the superior performance of our proposed solutions.
\end{abstract}


%
\section{Introduction}\label{sec:Introduction}
As a standard model for publishing and exchanging data on the Web, \emph{R}esource \emph{D}escription \emph{F}ramework (RDF) has been widely used in various applications to expose, share, and connect pieces of data on the Web. In RDF, data is represented as triples of the form $\langle$subject, property, object$\rangle$.  An RDF dataset can be naturally seen as a graph, where subjects and objects are vertices connected by named relationships (i.e., properties). SPARQL is a structured query language proposed by W3C to access RDF repository. As we know, answering a SPARQL query $Q$ is equivalent to finding subgraph matches of query graph $Q$ over an RDF graph $G$ \cite{DBLP:gStore}. Figures \ref{fig:exampleRDFGraph} and \ref{fig:ExampleSPARQL} show an RDF graph and a set of SPARQL query graphs used as the running example in this paper.

As RDF repositories increase in size, evaluating SPARQL queries is beyond the capacity of a single machine. For example, DBpedia, a project aiming to extract structured content from Wikipedia, consists of 2.46 billion RDF triples \cite{url:DBpedia}; according to the W3C, the numbers of triples in some commercial RDF datasets have been more than 1 trillion \cite{DBLP:Partout}. The large-scale of RDF data volume increases the demand of designing the high performance distributed RDF database system.

In distributed database design, the first issue is ``data fragmentation and allocation'' \cite{DBLP:distributedRDBMS}. We need to divide an RDF graph into several parts, called \emph{fragments}, and then distribute them among sites. One important issue during data fragmentation and allocation in a distributed system is how to reduce the communication cost between different fragments during distributed query evaluation (assuming different fragments are resident at different sites). To minimize the communication cost, many existing graph fragmentation strategies maximize the global goal (such as min-cut \cite{DBLP:metis}). However, evaluating a SPARQL query is a subgraph (homomorphism) match problem. The subgraph match computation often does not involve all vertices in graph $G$, and the communication cost of subgraph match computation depends on not only the RDF graph but also the query graph. In other words, subgraph match computation exhibits strong locality. There is no direct relation between minimizing the communication cost (in subgraph match computation) and maximizing the global goal. Hence, we propose a \emph{local pattern-based fragmentation} strategy in this paper, which can reduce the communication cost of subgraph match computation.

\nop{
There is no direct relation between minimizing the communication cost (in subgraph match computation) and maximizing the global goal (such as min-cut \cite{DBLP:metis}) in the first category approaches. Thus, the first category of graph fragmentation strategies are not suitable for SPARQL query problem. Instead, if a graph fragmentation can lead to few crossing matches, it can reduce the communication cost.
}

The intuition behind the local pattern-based fragmentation is as follows: if a query ``satisfies'' a local pattern and all its matches are in a single fragment, then the query can be evaluated on the single fragment and no communication cost is needed to answering the query. The key issue in local pattern-based fragmentation is how to define the ``local patterns''. Different from the existing methods, we consider the query workload-driven ``local pattern'' definition.

\subsection{Why Query Workload Matters ?}
The workload-driven distributed data fragmentation has been well studied in relational databases \cite{DBLP:distributedRDBMS}.  However, few RDF data fragmentation proposals consider the query workload except for \cite{DBLP:WARP,DBLP:Partout}. We will review these related papers in Section \ref{sec:related}. Here, we discuss why the query workloads is important for RDF data fragmentation.

We study one real SPAQRL query workload, the DBpedia query workload, which records 8,151,238 SPARQL queries issued in 14 days of 2012\footnote{http://aksw.org/Projects/DBPSB.html}. For this workload, if we set the minimum support threshold as 0.1$\%$ of the total number of queries, we mine 163 frequent subgraph patterns. The most surprising is that 97$\%$ query graphs are isomorphic to one of the 163 frequent subgraph patterns. Thus, if we use these frequent subgraph patterns as the basic fragmentation units, 97$\%$ SPARQL queries do not lead to communication cost, since their matches are resident at one fragment.

\subsection{Our Solution}
According to the above motivation, we propose a workload-driven data fragmentation for distributed RDF graph systems. Specifically, we
first mine frequent subgraph patterns, named \emph{frequent access patterns}, in the query workload. We treat these frequent access patterns as the \emph{implicit} schemas for the underlying RDF data. Then, we propose two fragmentation strategies based on these implicit schemas. We study the following technical issues in this paper.

\emph{\underline{Frequent Access Pattern Selection}.} Given a frequent access pattern, we build a fragment by collecting all its matches in the RDF graph. In this way, we can reduce the communication cost (i.e., improve query performance) if a SPARQL query satisfies the frequent access pattern. However, if we simply select all frequent access patterns as the implicit schemas, it may lead to expensive space cost due to the data replication, since different frequent access patterns may involve share the same edges. In other words, we have a tradeoff between performance gain and space cost during selecting frequent access patterns. We formalize the frequent access pattern selection problem (Section \ref{sec:PatternsSelection}) and prove that it is a NP-hard problem (Theorem \ref{theorem:Submodularity}). Thus, we propose a heuristic algorithm which can guarantee the data integrity and the approximation ratio (Theorem \ref{theorem:ApproximationRatio}). This algorithm also achieves the good performance (See experiments in Section \ref{sec:Experiment}).

\begin{figure}
\begin{center}
    \includegraphics[scale=0.35]{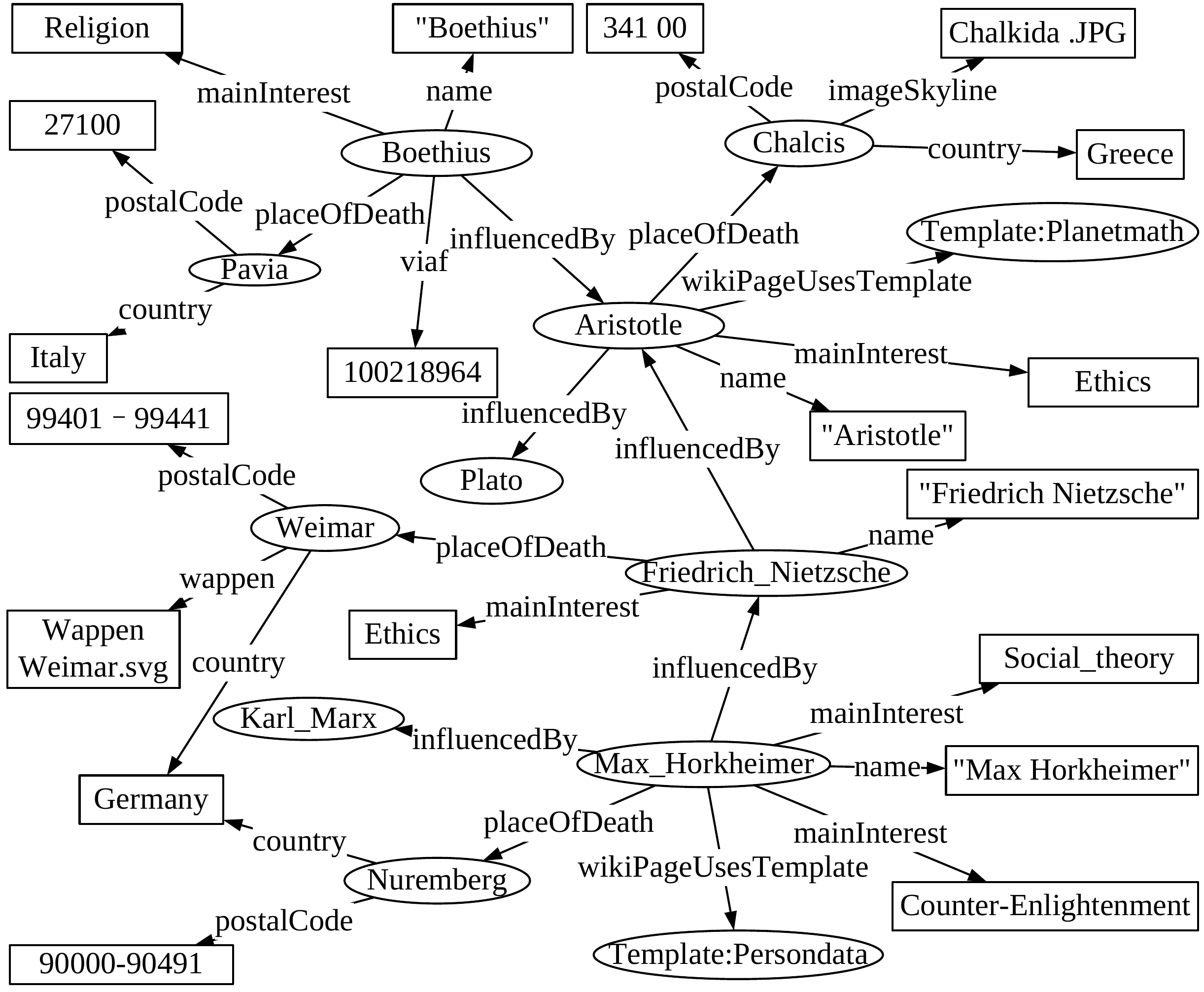}
    \vspace{-0.2in}
   \caption{Example RDF Graph}
   \label{fig:exampleRDFGraph}
\end{center}
\end{figure}

\begin{figure}
\begin{center}
    \includegraphics[scale=0.25]{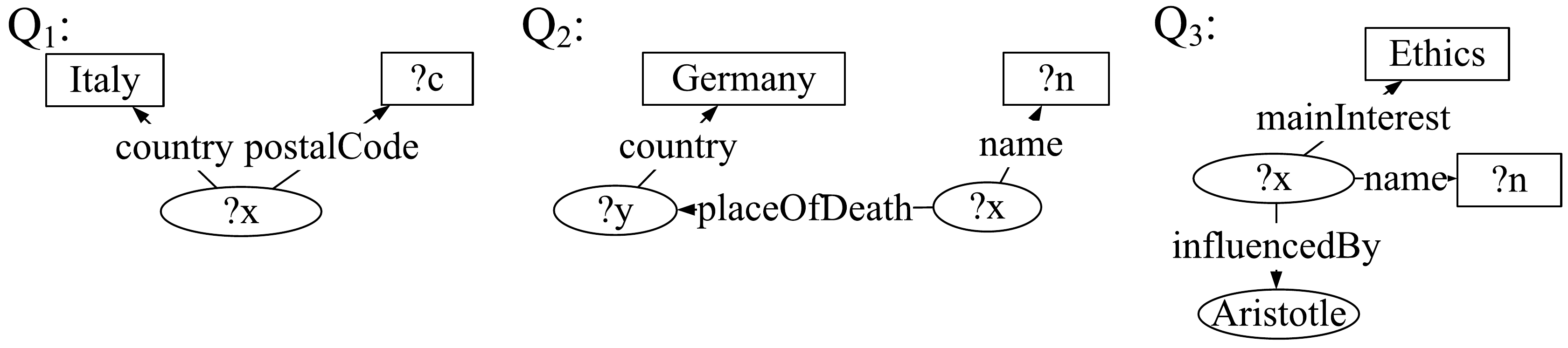}
    \vspace{-0.2in}
   \caption{Example SPARQL Query Graphs}
   \label{fig:ExampleSPARQL}
\end{center}
\end{figure}

\nop{
\emph{\underline{Vertical, Horizontal and Mixed Fragmentation}.} Based on the selected frequent access patterns (i.e., implicit schemas), we design three different fragmentation strategies, i.e, vertical, horizontal and mixed fragmentation. To perform the horizontal fragmentation over RDF graphs, we extend the concept of ``minterm predicate'' in distributed relational database systems to ``structural minterm predicate'' (see Section \ref{sec:HorizontalFragmentation}), which consider the structures of both RDF graphs and workloads. Different applications have different requirements, so we provide customizable options that can be used for different RDF graphs and SPARQL query workloads.
}

\emph{\underline{Vertical and Horizontal Fragmentation}.} Based on the selected frequent access patterns (i.e., implicit schemas), we design two fragmentation strategies, i.e, vertical and horizontal fragmentation. These two fragmentation strategies are adaptive to different query processing objectives. The objective of vertical fragmentation strategy is to improve the query throughout, and requires that all structures involved by one frequent access pattern should be placed to the same fragment. Instead, the horizontal fragmentation strategy distributes the structures involved by one frequent access pattern among different fragments to maximize the parallelism of query evaluation, namely, reducing the query response time for a single query. To perform the horizontal fragmentation over RDF graphs, we extend the concept of ``minterm predicate'' in \cite{DBLP:distributedRDBMS} to ``structural minterm predicate'' (see Section \ref{sec:HorizontalFragmentation}), which consider the structures of both RDF graphs and workloads. Different applications have different requirements, so we provide customizable options that can be used for different RDF graphs and SPARQL query workloads.

\emph{\underline{Query Decomposition}.} As we know, the query decomposition always depends on the fragmentation. In traditional vertical and horizontal fragmentation in RDBMS and XML, the query decomposition is unique, since there is no overlap between different fragments. As mentioned before, there are some data replications in our fragmentation strategies for RDF graphs. Thus, we may have multiple decomposition results for a query. A cost model driven selection is proposed in this paper.

The contributions of this paper can be summarized as follows:

\begin{itemize}
  \item We analyze the characteristics of the real SPARQL query workload and use the intrinsic similarities of queries in the workload to mine and select some frequent access patterns for distributed RDF data design. Although we prove that the problem of frequent access pattern selection is NP-hard, we propose a heuristic method to achieve the good performance.
  \item Based on the above scheme, we propose two fragmentation strategies, vertical and horizontal fragmentation, to divide the RDF graph into many fragments and a cost-aware allocation algorithm to distribute fragments among sites. The two fragmentation strategies provide customizable options that are adaptive to different applications.
  \item We propose a cost-aware query optimization method to decompose a SPARQL query and generate a distributed execution plan. With the decomposition results and execution plan, we can efficiently evaluate the SPARQL query.
  \item We do experiments over both real and synthetic RDF datasets and SPARQL query workloads to verify our methods.
\end{itemize}

\nop{
The rest of paper is organized as follows. First, we briefly introduce the background of this paper in Section \ref{sec:background}. Then, we give an overview of our methods in Section \ref{sec:Overview}. We formally define the \emph{frequent access patterns} and discuss how to select the optimal set of frequent access patterns in Section \ref{sec:FragmentPattern}. Then, we discuss how to use the frequent access patterns to divide the RDF graph in Section \ref{sec:Fragmentation} and how to allocate all fragments to different sites in Section \ref{sec:Allocation}. How to process queries are discussed in Section \ref{sec:QueryProcessing}. Experiments on real RDF datasets and SPARQL query workload are presented in Section \ref{sec:Experiment}. Finally, we discuss the related work in Section \ref{sec:related}, and conclude in Section \ref{sec:Conclusion}.
}

\section{Preliminaries}\label{sec:background}
In this section, we review the terminologies used in this paper and formally define the problem to be addressed.

\subsection{RDF and SPARQL}\label{sec:background-RDF-SPARQL}
RDF data can be represented as a graph according to the following definition.

\begin{definition}\label{def:graph} \textbf{(RDF Graph)}
An RDF graph is denoted as $G=\{V(G),$ $E(G),L \}$, where (1) $V(G)$ is a set of vertices that correspond to all subjects and objects in RDF data; (2) $E(G) \subseteq V(G) \times V(G)$ is a set of directed edges that correspond to all triples in RDF data; and (3) $L$ is a set of edge labels. For each edge $e \in E(G)$, its edge label is its corresponding property.
\end{definition}

Similarly, a SPARQL query can also be represented as a query graph $Q$. For simplicity, we ignore FILTER statements in SPARQL syntax in this paper.
\begin{definition}\label{def:query}\textbf{(SPARQL Query)}
A \emph{SPARQL query} is denoted as $Q=\{V(Q),$ $E(Q), L^{\prime}\}$, where (1) $V(Q) \subseteq V(G)\cup V_{Var}$ is a set of vertices, where $V(G)$ denotes vertices in RDF graph $G$ and $V_{Var}$ is a set of variables; (2) $E(Q) \subseteq V(Q) \times V(Q)$ is a set of edges in $Q$; and (3) $L^{\prime}$ is also a set of edge labels, and each edge $e$ in $E(Q)$ either has an edge label in $L$ (i.e., property) or the edge label is a variable.
\end{definition}

In this paper, we assume that $Q$ is a connected graph; otherwise, all connected components of Q are considered separately. Given a SPARQL query $Q$ over RDF graph $G$, a SPARQL match is a subgraph of G that is homomorphic to $Q$ \cite{DBLP:gStore}. Thus, answering a SPARQL query is equivalent to finding all subgraph matches of $Q$ over RDF graph $G$. The set of all matches for $Q$ over $G$ is denoted as $\llbracket Q \rrbracket_G$

In this work, we study a query workload-driven fragmentation. A query workload $\mathcal{Q}=\{Q_1,Q_2,...,Q_q\}$ is a set of queries that users input in a given period.

\subsection{Fragmentation $\&$ Allocation}
In this paper, we study an efficient distributed SPARQL query engine. There are many issues related to distributed database system design, but, the focus of this work is ``data fragmentation and allocation'' for RDF repository. We formalize two important problems as follows.

\begin{definition}\textbf{(Fragmentation)} Given an RDF graph $G$, a \emph{fragmentation} $\mathcal{F}$ of $G$ is a set of  graphs $\mathcal{F}=\{F_1, ..., F_n\}$ such that: (1) each $F_i$ is a subgraph of $G$ and called as a \emph{fragment} of RDF graph $G$; (2) $E(F_1)  \cup ... \cup E(F_n)  = E(G)$; and (3) $V(F_1)  \cup ... \cup V(F_n)  = V(G)$, where $E(F_i)$ and $V(F_i)$ denote the edges and vertices in $F_i$ ($i=1,..,n$).
\end{definition}

In our work, we allow the overlaps between different fragments. Given a fragmentation $\mathcal{F}$, the next issue is how to distribute these fragments among different sites (i.e., computing nodes). This is called \emph{allocation}.

\begin{definition}\textbf{(Allocation)} Given a fragmentation $\mathcal{F}=\{F_1, ...,$ $F_n\}$ over an RDF graph $G$ and a set of sites $\mathcal{S} = \{S_1,S_2,...,S_m\}$ (usually $m<n$), an \emph{allocation} $\mathcal{A}=\{A_1, ..., A_m\}$ of fragments in $\mathcal{F}$ to $\mathcal{S}$ is a partitioning of $\mathcal{F}$ such that (1) $A_j\subseteq \mathcal{F}$, where $1 \le j\le m$; (2) $A_{j_1} \cap A_{j_2} = \emptyset$, where $1 \le j_1 \ne j_2 \le m$; (3) $A_1  \cup ... \cup A_m  = \mathcal{F}$; and (4) All fragments in $A_j$ are stored at site $S_j$, where $1 \le j\le m$.
\end{definition}

Given an RDF graph $G$, a query workload $\mathcal{Q}$ and a distributed system consisting of sites $\mathcal{S}$, the goal of this paper is to first decompose $G$ into a fragmentation $\mathcal{F}$ and then finding the allocation $\mathcal{A}$ of $\mathcal{F}$ to $\mathcal{S}$.

\section{Overview}\label{sec:Overview}
This paper studies a SPARQL query workload-driven data fragmentation and allocation problem. Some observations on the real query workload tell us that some RDF properties have few access frequencies. For example, few users input queries contain the properties like $imageSkyline$ and $wappen$ in Figure \ref{fig:exampleRDFGraph}. As well, the classical distributed database design suggests a ``80/20'' rule, meaning the active ``20$\%$'' of query patterns account for ``80$\%$'' of the total query input \cite{DBLP:DatabaseDesign}. Therefore, we divide the whole RDF repository into two parts: ``hot graph'' and ``cold graph'' as follows.

\begin{definition} \textbf{(Infrequent and Frequent Property)} Given a query workload $\mathcal{Q}=\{Q_1,...Q_n\}$, if a property $p$ occurs in less than $\theta$ queries in $\mathcal{Q}$, where $\theta$ is an user specified parameter, $p$ is an \emph{infrequent property}; otherwise, $p$ is a \emph{frequent property}.
\end{definition}

\begin{definition} \textbf{(Hot and Cold Graphs)} Given an edge $e=\overrightarrow{u_iu_j}\in E(G)$ with property $p$, if property $p$ is a frequent property, $e$ is a hot edge; otherwise, $e$ is a cold edge.

Given an RDF graph $G$, it is divided into two parts: \emph{hot graph} $H$ and \emph{cold graph} $C$, where $H$ consists of all hot edges and $C$ consists of all cold edges.
\end{definition}

The goal of this work is how to partition ``hot graph'' to achieve performance improvement. We regard the cold graph as a ``black block''. The cold graph does not overlap to the hot graph, since the cold graph contains different edges with different kinds of properties from the hot graph. Any existing approach can be utilized for the cold graph. We only consider the cold graph in the SPARQL query processing (Section \ref{sec:QueryProcessing}), since some queries may involve ``infrequent'' properties. Moreover, both the cold graph and the hot graph may be disconnected.

Figure \ref{fig:SystemArchitecture} illustrates our system architecture. In the offline phase, we mine the \emph{frequent access patterns} (see Section \ref{sec:FragmentPattern}) in the workload. Each frequent access pattern can correspond to one or more fragments. Generating a fragment from all matches of a frequent access pattern make many queries be answered efficiently without cross-fragments joins, while it may also replicate some hot edges and increase the space cost. Thus, we should select an appropriate subset of frequent access patterns to balance the efficiency and the space cost. Since we find out that selecting an appropriate set of patterns is a NP-hard problem (Section \ref{sec:PatternsSelection}), we propose a heuristic pattern selection solution while guaranteeing both the data integrity and the approximation ratio. Based on these selected frequent access patterns, we study two different data fragmentation strategies, i.e., vertical and horizontal fragmentation (Section \ref{sec:Fragmentation}). The vertical fragmentation is to improve the query throughput, and the horizontal fragmentation is to reduce a single query's response time. Fragments are distributed among different sites. Meanwhile, we maintain the metadata in a data dictionary.

In the online phase, we study how to decompose a query into several subqueries on different fragments and generate an efficient execution plan. A cost model for guiding decomposition is proposed (Section \ref{sec:QueryDecomposition}). Finally, we execute the plan and return the matches of the query (Section \ref{sec:QueryOptimization}).

\begin{figure}
\begin{center}
    \includegraphics[scale=0.3]{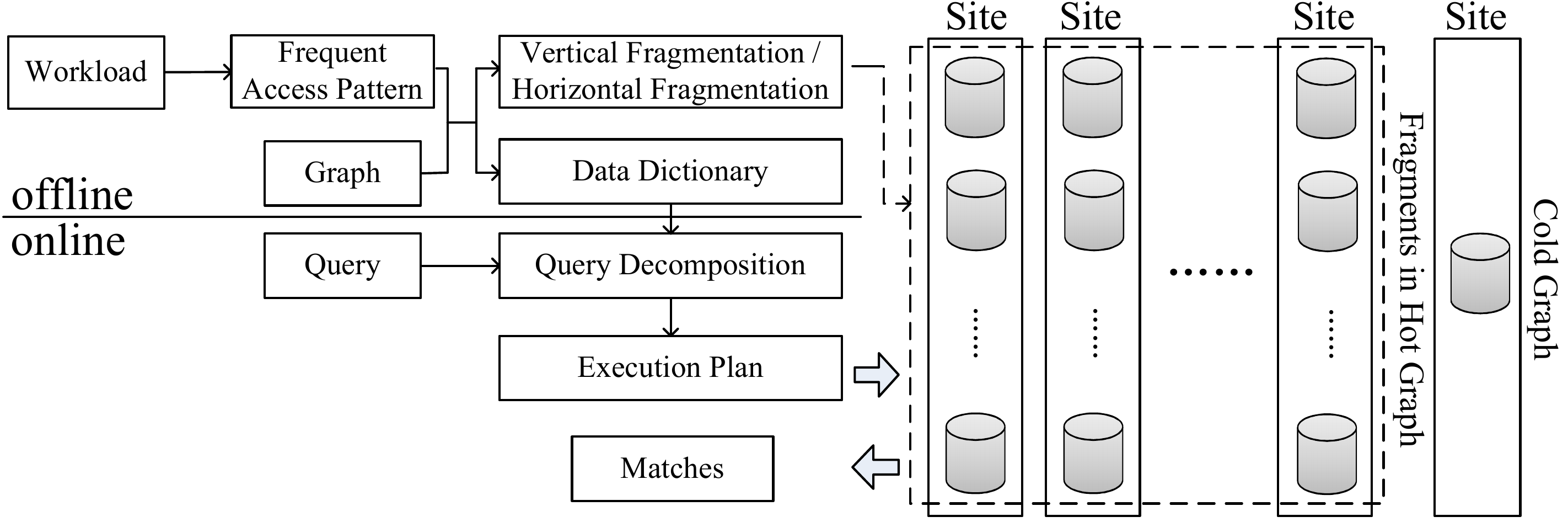}
   \caption{System Architecture}
   \label{fig:SystemArchitecture}
\end{center}
\end{figure}

\section{Frequent Access Patterns}\label{sec:FragmentPattern}
As mentioned before, we believe that a query often contains some patterns in the previously issued queries, so we mine some patterns with high access frequencies and use these patterns as the fragmentation units. Then, if a query $Q$ can be decomposed to some subgraphs isomorphic to the frequent access patterns, $Q$ can be answered while avoiding some joins across multiple fragments.

Before we mine frequent access patterns, we first normalize the query graphs in the workload to avoid overfitting. For each SPARQL query, we remove all constants (strings and URIs) at subjects and objects and replace them with variables. The FILTER expressions are also removed. By doing this, we extract a general representation of a SPARQL query from the workload. Figure \ref{fig:generalizedSPARQL} shows the generalized query graphs of query graphs in Figure \ref{fig:ExampleSPARQL}. We assume that the generalized query in Figure \ref{fig:generalizedSPARQL} graphs are also frequent access patterns.

\begin{figure}
\begin{center}
    \includegraphics[scale=0.25]{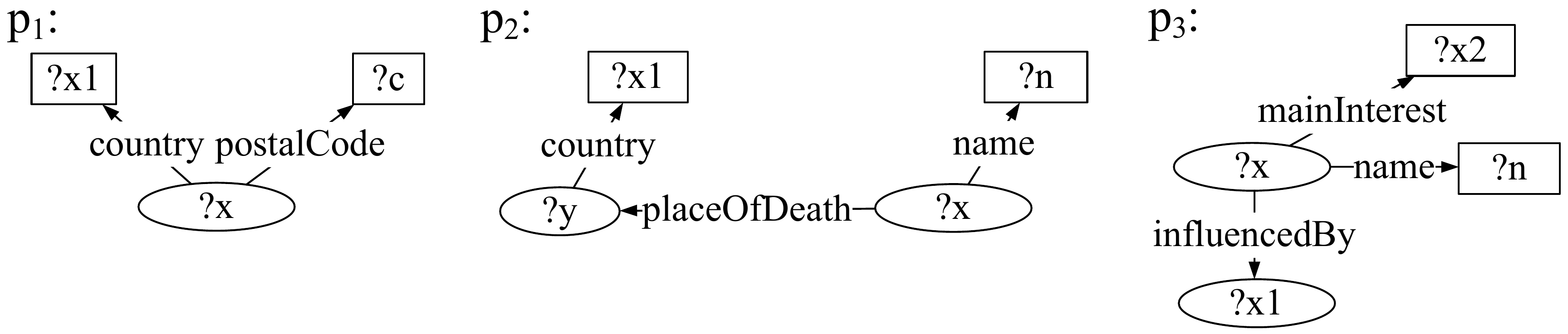}
   \caption{Example Frequent Access Patterns}
   \label{fig:generalizedSPARQL}
\end{center}
\end{figure}

To mine patterns with high access frequencies, we need to first count the number of queries in the workload where a pattern $p$ is a subgraph. We define the \emph{frequent access pattern usage value} to record the access frequencies of the frequent access patterns.

\begin{definition}\label{def:PropertyUsageValue}
\textbf{(Frequent Access Pattern Usage Value)}
Given a SPARQL query $Q$ and a frequent access pattern $p$, we associate a \emph{frequent access pattern usage value}, denoted as $use(Q, p)$, and defined as follows:
$$ use(Q, p)=\left\{
\begin{aligned}
1  &      &  if\ pattern\ p\ is\ a\ subgraph\ of\ Q\\
0  &      &  otherwise
\end{aligned}
\right.
$$
\end{definition}

Then, given a workload $\mathcal{Q}=\{Q_1,Q_2,..., Q_q\}$ and a pattern $p$, we define the \emph{access frequency}, $acc(p)$, as the number of queries in $\mathcal{Q}$ where a pattern $p$ is a subgraph.
\[acc(p)=\sum_{k=1}^{q} {use(Q_k,p)}\]
A pattern $p$ is \emph{frequent access pattern} if its access frequency is no less than a threshold, $minSup$. \nop{As one can see, frequent access pattern is a relative concept. Whether a graph is frequent depends on the setting of $minSup$.}

The frequent access patterns can be easily generated by existing frequent graph mining algorithms \cite{DBLP:Gaston}. Given a workload of SPARQL queries $\mathcal{Q} = \{Q_1,Q_2,...,$ $Q_q\}$ in a given period, we denote the set of frequent access patterns that we find as ${P}=\{p_1,p_2,...,p_x\}$. In practice, the size of ${P}$ is often limited. For example, if we set $minSup$ as $0.1\%$ of the total access frequency, there are only $163$ frequent access patterns for DBPedia.

\subsection{Frequent Access Pattern Selection}\label{sec:PatternsSelection}
Obviously, it is not necessary to generate fragments from all frequent access patterns due to high space cost. For two similar frequent access patterns $p$ and $p^\prime$, if they are contained by similar queries of the workload, then selecting both $p$ and $p^\prime$ for building fragments will not be able to provide more information than selecting one of $p$ and $p^\prime$. Hence, it is often sufficient to only select a subset of all frequent access patterns to generate fragments.

To select a subset of all frequent access patterns, there are two factors that we should consider.
\begin{enumerate}
\item (\emph{Hitting the Whole Workload}) We should select frequent access patterns to hit the query workload as much as possible. This is because that when we select a frequent access pattern to generate a fragment, all queries isomorphic to this pattern can be answered directly, which improve the efficiency.
  \item (\emph{Satisfying the Storage Constraint}) The total storage of the system in real applications is limited, so selecting too many frequent access patterns is not desirable.
\end{enumerate}

The above two factors contradict each other. Hitting the whole workload requires to select as many frequent access patterns as possible, while the storage constraint requires to select not too many frequent access patterns. There should be a tradeoff between the two factors.
In the following, we propose a cost model to combine these two factors for selecting a set of frequent access patterns.

\subsubsection{Hitting the Whole Workload} If a fragment is generated from the graph induced by matches of a frequent access pattern, then evaluating all queries containing the pattern can be speeded up by using this fragment. The more queries a frequent access pattern hits, the more gains we obtain during query processing. Therefore, the \emph{benefit} of selecting a frequent access pattern to generate its corresponding fragment should be defined based on the number of queries that the frequent access pattern hits.

In addition, if two similar frequent access patterns are contained by the same set of queries in the workload, it is probably wise to include only one of them. Generally speaking, among similar frequent access patterns contained by the same number of queries, it is often sufficient to materialize only the largest frequent access pattern. That is to say, if $p^\prime$, a subgraph of $p$, is contained by the same set of queries as $p$, $p$ is more beneficial than $p^\prime$ to be selected as building fragments. This is because that if we select the larger pattern, a query is more probable to be decomposed to fewer number of subqueries during query processing. Fewer subqueries can avoid some distributed joins, which can improve the efficiency of query processing.

The above observation implies that larger frequent access patterns are more beneficial to be selected as building fragments. This above criterion on the selection of frequent access patterns is formally defined as \emph{size-increasing benefit}.

\begin{definition}\label{def:PatternBenefit}\textbf{(Size-increasing Benefit)}
Given a frequent access pattern $p$, the benefit of selecting $p$ for hitting the query ${Q}$, $Benefit(p, {Q})$, is denoted as follows.
\[Benefit(p, {Q})={|E(p)|} \times use(Q, p)\]
\end{definition}

Furthermore, a query in the workload may contain multiple selected frequent access patterns. This means that the query can be decomposed into multiple sets of subqueries if we evaluate the query. Each set of subqueries can map to an execution plan. Since only one execution plan is finally selected to evaluate the query, a query in the workload should only be limited to contribute to the benefits of some particular frequent access patterns once. Based on this observation, we limit a query to only contribute the largest frequent access pattern that the query contains.

\begin{definition}\label{def:Benefit}\textbf{(Benefit of a Frequent Access Pattern Set)}
Given a set of frequent access patterns ${P}^\prime \subseteq {P}$, the benefit of selection of ${P}^\prime$ over the workload $\mathcal{Q}$ is the sum of the maximum benefit of its frequent access patterns over $\mathcal{Q}$.
\[Benefit(P^\prime, \mathcal{Q}) =  \sum_{Q\in \mathcal{Q}} \max_{p \in P^\prime} \{Benefit(p, {Q})\}\]
\end{definition}

\subsubsection{Satisfying the Storage Constraint}
Furthermore, the total storage of the system in real applications is limited, so selecting too many frequent access patterns is not desirable. The selection of frequent access patterns should meet some constraints. When the size of all fragments is larger than the storage constraint, we cannot further select any more frequent access patterns. We normalize the storage capacity of the system to a value $SC$. Then, we have the constraint as:
\[\sum\limits_{p\in {P}^\prime}|\llbracket p \rrbracket|\times |E(p)| \le SC\]
Here, we assume that $SC$ is larger than the number of edges in the hot graph, so each hot edge can have at least one copy. This assumption guarantees the completeness of the RDF graph.

\nop{
If two frequent access patterns share a common substructure, their corresponding fragments have some redundant edges. This increase the space cost. Hence, we define the cost of the selection of a set of frequent access patterns as the number of redundant edges that these patterns share.

Given a set of frequent access patterns ${P}^\prime \subseteq {P}$, the cost $Cost(P^\prime)$ of the selection of $P^\prime$ is defined as the number of replicated edges for $P^\prime$.
\[Cost(P^\prime)=\sum\limits_{p\in {P}^\prime}|E(\llbracket p \rrbracket_G)|-|\bigcup\limits_{p\in P^\prime}E(\llbracket p \rrbracket_G)|\]
where $\llbracket p \rrbracket_G$ denotes all substructures that matches $p$, and $E(\llbracket p \rrbracket_G)$ is the set of edges in any substructures that matches $p$
}

\subsubsection{Combining the Two Factors}\label{sec:CombiningFactors}
\nop{
We combine the above two factors together to define the benefit function of the selection of a set of frequent access patterns $P^\prime$.
\begin{equation}\label{equation:benefit}
Benefit(P^\prime)= Gain(P^\prime) - \gamma \times Cost(P^\prime)
\end{equation}
where $\gamma$ is a parameter not smaller than 0. We will discuss the settings of parameters $\gamma$ in Section \ref{sec:expsettingalpha}.
}

Then, our optimization objective is to maximize the benefit subject to the storage constraint. We can prove that this benefit function (Definition \ref{def:Benefit}) is submodular as follows, so this problem is NP-hard.

\begin{theorem}\label{theorem:Submodularity} Finding a set of frequent access patterns with the largest benefit while subject to the storage constraint is NP-hard.
\end{theorem}
\begin{proof} Here, we prove that the benefit function $Benefit(P^\prime,\mathcal{Q})=\sum_{Q\in \mathcal{Q}} \max_{p \in P^\prime} \{{|E(p)|} \times use(Q, p)\}$ is submodular. In other words, for every $P_1\subseteq P_2$ and a frequent access pattern $p\notin P_2$, we need to prove that $\triangle_{Benefit}(p|P_1)\ge \triangle_{Benefit}(p|P_2)$.

For pattern $p$, we assume that $\mathcal{Q}^\prime$ is the set of queries containing $p$ in the workload. There are three kinds of queries in $\mathcal{Q}^\prime$: the set $\mathcal{Q}_1$ of queries not containing any patterns in $P_2$, the set $\mathcal{Q}_2$ of queries containing patterns in $(P_2-P_1)$, and the set $\mathcal{Q}_3$ of queries only containing patterns in $P_1$.

Since any query in $\mathcal{Q}_1$ and $\mathcal{Q}_3$ does not concern patterns in $(P_2-P_1)$, $Benefit(\{p\}\cup P_1, \mathcal{Q}_1\cup \mathcal{Q}_3)=Benefit(\{p\}\cup P_2, \mathcal{Q}_1\cup \mathcal{Q}_3)$. Hence, the marginal gains of $p$ for $P_1$ and $P_2$ over $\mathcal{Q}_1$ and $\mathcal{Q}_3$ are the same.

For $\mathcal{Q}_2$, $\triangle_{Benefit}(p|P_1)> \triangle_{Benefit}(p|P_2)$, if there exist at least one query $Q^*$ meeting all the two following conditions: 1) the largest pattern contained by $Q^*$ over $P_2$ is in $(P_2-P_1)$ and has larger size than $p$; 2) the largest pattern contained by $Q^*$ over $P_1$ has smaller size than $p$. The above two conditions mean that $p$ can only increase the benefit of $P_1$ over $\mathcal{Q}_2$ but not the benefit of $P_2$ over $\mathcal{Q}_2$. Otherwise, for $\mathcal{Q}_2$, $\triangle_{Benefit}(p|P_1)= \triangle_{Benefit}(p|P_2)$.

In conclusion, $\triangle_{Benefit}(p|P_1)\ge \triangle_{Benefit}(p|P_2)$ and the function $Benefit(P^\prime, \mathcal{Q})$ is submodular. Since the problem of maximizing submodular functions is NP-hard \cite{DBLP:submodular}, the problem is NP-hard.
\end{proof}

\subsubsection{Our Solution}\label{sec:PatternSelectionSolution}
As proved in Theorem \ref{theorem:Submodularity}, frequent access pattern selection is NP-complete problem. We propose a greedy algorithm as outlined in Algorithm \ref{theorem:Submodularity}. Note that, to guarantee data integrity of distributed RDF data fragmentation, each hot edge should be contained in at least one fragment. Hence, we initialize a pattern of one edge for each frequent property and compute out its corresponding fragment (Line 3-6).

\nop{Note that, when we select the patterns of one edge, the benefit of the selection always increases. This is because that one edge only has one property and different patterns with one edge cannot share the same edge. }

After we select all patterns with one edge, we enumerate all feasible frequent access pattern sets containing one pattern of more than one edge. Let $P_1$ be a feasible set of cardinality one that has the largest benefit (Line 7). Then, we iteratively select one of the remaining frequent access patterns $p^\prime$ to maximize the value of $\frac{{Benefit}(\{p^\prime\}\cup P^\prime,\mathcal{Q}) - {Benefit}( P^\prime,\mathcal{Q})}{|E(\llbracket p^\prime \rrbracket_G)|}$ until we meet the storage constraint or cannot find a frequent access pattern to increase the benefit (Line 8-14). Let $P_2$ be the solution obtained in the iterative phase. Finally, the algorithm outputs $P^\prime\cup P_1$ if ${Benefit}( P^\prime\cup P_1,\mathcal{Q})\ge {Benefit}( P^\prime\cup P_2,\mathcal{Q})$ and $P^\prime\cup P_2$ otherwise (Line 15-17).

\begin{algorithm}[h] \label{alg:PatternSelectionSolution}
\caption{Frequent Access Pattern Selection Algorithm}
\small
\KwIn{A set of frequent access patterns ${P}=\{p_1,p_2,...,p_x\}$}
\KwOut{ A set ${P^\prime}\subseteq P$ to generate fragments}

${P^\prime}\gets \emptyset$;\\
$TotalSize\gets 0$;\\
\For{each $p\in P$ and $p$ has only one edge}
{
    ${P^\prime}\gets {P^\prime}\cup \{p\}$;\\
    ${P}\gets {P}- \{p\}$;\\
    $TotalSize\gets TotalSize + |E(\llbracket p \rrbracket_G)|$;\\
}
$P_1 \gets argmax\{\frac{{Benefit}( \{p_i\},\mathcal{Q})}{|E(\llbracket p_i \rrbracket_G)|}:p_i\in P, |E(\llbracket p_i \rrbracket_G)| + TotalSize \le SC \land |E(p_i)|>1\}$;\\
${P_2}\gets \emptyset$;\\
$TotalSize^\prime \gets 0$;\\
\While{$TotalSize^\prime \le SC-TotalSize$}
{
    Find the frequent access pattern $p^\prime\in P-P^\prime$ with the largest additional value of $\frac{{Benefit}(\{p^\prime\}\cup P^\prime,\mathcal{Q}) - {Benefit}( P^\prime,\mathcal{Q})}{|E(\llbracket p^\prime \rrbracket_G)|}$;\\
    ${P_2}\gets {P_2}\cup \{p^\prime\}$;\\
    ${P}\gets {P}- \{p^\prime\}$;\\
    $TotalSize^\prime\gets TotalSize^\prime + |E(\llbracket p^\prime \rrbracket_G)|$;\\
}
\If{${Benefit}( P^\prime\cup P_1,\mathcal{Q})\ge {Benefit}( P^\prime\cup P_2,\mathcal{Q})$}{
    Return ${P^\prime\cup P_1}$;
}
Return $P^\prime\cup P_2$;
\end{algorithm}

\begin{theorem}\label{theorem:ApproximationRatio} Algorithm \ref{alg:PatternSelectionSolution} obtains a set of frequent access patterns of benefit at least $\min\{\frac{1}{(\max_{p\in P}{|E(p)|})},\frac{1}{2}(1-\frac{1}{e})\}$ times the value of an optimal solution.
\end{theorem}
\begin{proof}
There are two parts in Algorithm \ref{alg:PatternSelectionSolution}: initialization and greedy selection of frequent access patterns.

For initialization (Line 3-6 in Algorithm \ref{alg:PatternSelectionSolution}), all selected patterns only contain one edge, so $|E(p)| = 1$. Therefore, the benefit of patterns only having one edge of a frequent property is $\sum_{Q\in \mathcal{Q}} \max_{p \in P^\prime} \{ 1 \times use(Q, p)\}$. Since the hot edges hit almost all queries in the workload, $\sum_{Q\in \mathcal{Q}} \max_{p \in P^\prime} \{ 1 \times use(Q, p)\}$ is approximately equal to the size of the workload, $|\mathcal{Q}|$. On the other hand, in the worst case, the optimal solution is that all queries in the workload contain the largest frequent access pattern. Then, the benefit of the optimal solution is $\sum_{Q\in \mathcal{Q}} \{ |E(p_{max})| \times use(Q, p)\}$, where $p_{max}$ is the frequent pattern with the largest size. Hence, the benefit of the selected patterns in the initial phase is at least $\frac{1}{(\max_{p\in P}{|E(p)|})}$ of the optimal benefit.

For the phase of greedily selecting frequent access patterns (Line 7-14 in Algorithm \ref{alg:PatternSelectionSolution}), since the problem of selecting the optimal set of frequent access patterns is a problem of maximizing a submodular set function subject to a knapsack constraint as discussed in Theorem \ref{theorem:Submodularity}, we directly apply the greedy algorithm in \cite{DBLP:journals/corr/IyerB13a} to iteratively select frequent access patterns. \cite{DBLP:journals/corr/IyerB13a} proves that the worst-case performance guarantee of the greedy algorithm is $\frac{1}{2}(1-\frac{1}{e})$, so the benefit of the selected patterns in this phase is at least $\frac{1}{2}(1-\frac{1}{e})$ of the optimal benefit.

In summary, the final performance guarantee of our algorithm is $\min\{\frac{1}{(\max_{p\in P}{|E(p)|})},\frac{1}{2}(1-\frac{1}{e})\}$.
\end{proof}

\section{Fragmentation}\label{sec:Fragmentation}
In this section, we present two fragmentation strategies: vertical and horizontal.

\subsection{Vertical Fragmentation}\label{sec:VerticalFragmentation}
For vertical fragmentation, we put matches homomorphic to the same frequent access pattern into the same fragment. Because a query graph often only contains a few frequent access patterns and matches of one frequent access pattern are put together, other irrelevant fragments can be filtered out during query evaluation and only sites stored relevant fragments need to be accessed to find matches. Filtering out irrelevant fragments can improve the query performance. Furthermore, sites not storing relevant fragments can be used to evaluate other queries in parallel, which improves the total throughput of the system. In summary, the vertical fragmentation strategy utilizes the locality of SPARQL queries to improve both query response time and throughput. Experimental results in Section \ref{sec:Experiment} also confirm the above argument.

Given a frequent access pattern $p$, it can then be transformed into a SPARQL query, resulting in a vertical fragment of the RDF graph. We use the results $\llbracket p \rrbracket_G$ of a selection operation based on $p$ to generate a vertical fragment. All vertical fragments generated from our selected frequent access patterns construct a vertical fragmentation. Given a set of frequent access patterns $P$, we formally define its corresponding vertical fragmentation over an RDF graph $G$ as follows.

\begin{definition}\label{def:verticalfragmentation} \textbf{(Vertical Fragmentation)}
Given an RDF graph $G$ and a frequent access pattern $p$, a \emph{vertical fragment} $F$ generated from $p$ is defined as $F=\{V(F),$ $E(F),L^{\prime\prime} \}$, where (1) $V(F)\subseteq V(G)$ is the set of vertices occurring in $\llbracket p \rrbracket_G$; (2) $E(F) \subseteq E(G)$ is the set of edges occurring in $\llbracket p \rrbracket_G$; and (3) $L^{\prime\prime} \subseteq L$ is the set of edge labels occurring in $\llbracket p \rrbracket_G$.

Then, given a set of frequent access patterns ${P}=\{p_1,p_2,...,p_x\}$, the corresponding \emph{vertical fragmentation} is $\mathcal{F}=\{F_i|0\le i \le x$ and $F_i$ is the vertical fragment generated from $p_i$.$\}$
\end{definition}

\begin{example}
Given the frequent access pattern $p_3$ in Figure \ref{fig:generalizedSPARQL}, Figure \ref{fig:ExampleVerticalFragment} shows the corresponding vertical fragment.
\end{example}

\begin{figure}[h]
\begin{center}
    \includegraphics[scale=0.31]{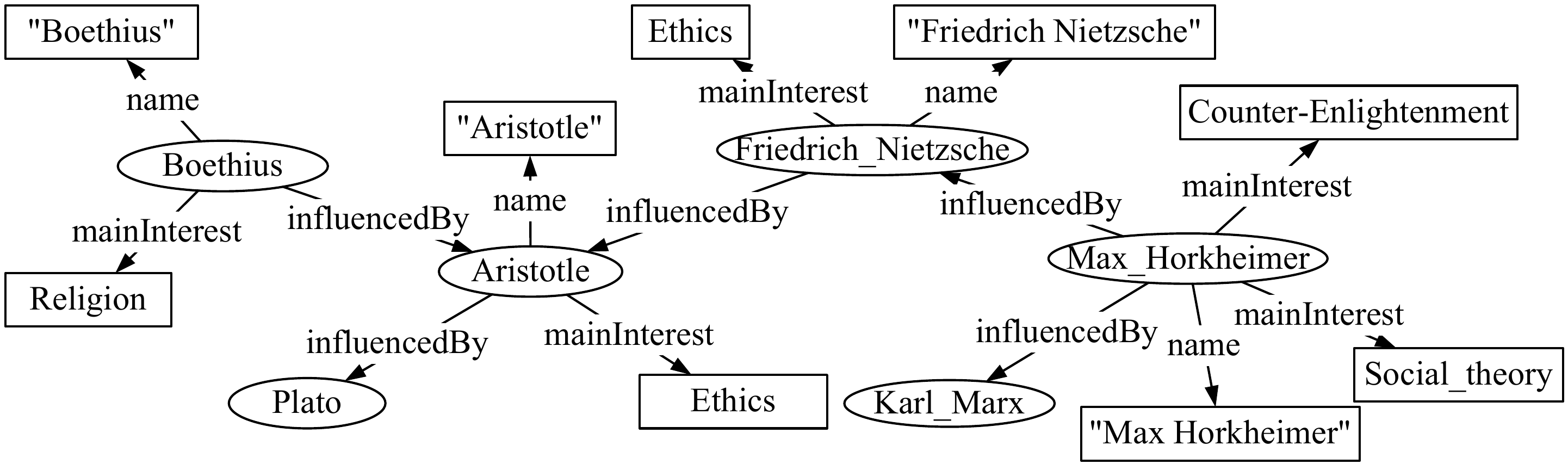}
    \vspace{-0.3in}
   \caption{Example Vertical Fragment}
   \label{fig:ExampleVerticalFragment}
\end{center}
\end{figure}
\vspace{-0.1in}

\subsection{Horizontal Fragmentation}\label{sec:HorizontalFragmentation}
For horizontal fragmentation, we put matches of one frequent access pattern into the different fragments and distribute them among different sites. Then, a query may involve many fragments and each fragment has a few matches. The size of a fragment is often much smaller than the size of the whole data, so finding matches of a query over a fragment explores smaller search space than finding matches over the whole data. If the fragments involved by a query are allocated to different sites, then each site finds a few matches over some fragments with the smaller size than the whole data. This strategy is to utilize the parallelism of clusters of sites to reduce the query response time. The above argument is also confirmed by the experimental results in Section \ref{sec:Experiment}.

In this section, we extend the concepts of \emph{simple predicate} and \emph{minterm predicate} originally developed for relational systems \cite{DBLP:distributedRDBMS} to divide the RDF graph horizontally.

\subsubsection{Structural Minterm Predicate }\label{sec:MintermPredicate}

First, we define the structural simple predicate. Each structural simple predicate corresponds to a frequent access pattern with a single (in)equality. Given a frequent access pattern $p$ with variables set $\{var_1,var_2,...,var_n\}$, a structural simple predicate $sp$ defined on $D$ has the following form.
\[sp:p({var_i})\ \theta \ Value \]
where $\theta \in \{=, \ne\}$ and $Value$ is a constant constraint for $var_i$ chosen from a query containing $p$ in $\mathcal{Q}$.

\begin{example}\label{example:SimplePredicates}
Let us consider the query graph $Q_3$ in Figure \ref{fig:ExampleSPARQL} and its corresponding frequent access pattern $p_3$ in Figure \ref{fig:generalizedSPARQL}. We can generate four structural simple predicates:
(1). $sp_1:p_3(?x1)=Aristotle$;
(2). $sp_2:p_3(?x1)\ne Aristotle$;
(3). $sp_3:p_3(?x2)=Ethics$;
(4). $sp_4:p_3(?x2)\ne Ethics$.
\end{example}

Then, we define the structural minterm predicate as the conjunction of structural simple predicates of the same frequent access pattern. We can obtain all structural minterm predicates by enumerating all possible combinations of structural simple predicates. Given a set of structural simple predicates $SP = \{sp_1, sp_2,...,, sp_y\}$ for frequent access pattern $p$, the set of structural minterm predicates $M = \{mp_1, mp_2,...,mp_z\}$ for $p$ is defined as follows.
\[M=\{mp_i| \bigwedge_{sp_k\in SP} sp^*_{k}, 1\le k \le y\} \]
where $sp^*_{k}= sp_k$ or $sp^*_{k}= \lnot sp_k$. So each structural simple predicate can occur in a structural minterm predicate either in its natural form or its negated form.

Similar to the frequent access pattern, we can also define the \emph{structural minterm predicate usage value} and \emph{access frequency} to record the access frequency of a structural minterm predicate. We can prune the minterm predicates with small access frequencies.

\begin{definition}\label{def:PropertyUsageValue}
\textbf{(Structural Minterm Predicate Usage Value)}
Given a SPARQL query $Q$ and a structural minterm predicate $mp$, we associate a \emph{structural minterm predicate usage value}, denoted as $use(Q, mp)$, and defined as follows:
$$ use(Q, mp)=\left\{
\begin{aligned}
1  &      &  if\ predicate\ mp\ is\ a\ subgraph\ of\ Q\\
0  &      &  otherwise
\end{aligned}
\right.
$$
\end{definition}

Then, given a set of SPARQL queries $\mathcal{Q} = \{Q_1,Q_2,...,Q_q\}$, we define the \emph{access frequency} of a structural minterm predicate $mp$ as follows.
\[acc(mp) = \sum _{k=1}^{k=q} {use({Q_k}, mp)} \]

In practice, there may exist many minterm predicates. It is too expensive to enumerate all minterm predicates. Therefore, we prune some minterm predicates with too small access frequencies.

Given a structural minterm predicate $mp$, it can then be transformed into SPARQL queries, resulting in a horizontal fragment of the RDF graph. We use the results $\llbracket mp \rrbracket_G$ of a selection operation based on $mp$ to generate a horizontal fragment. All horizontal fragments generated from the structural minterm predicates that we obtain construct a horizontal fragmentation. Given a set of minterm predicates $M$, we formally define its corresponding horizontal fragmentation over an RDF graph $G$ as follows.

\begin{definition}\label{def:horizontalfragmentation} \textbf{(Horizontal Fragmentation)}
Given an RDF graph $G$ and a structural minterm predicate $mp$, a \emph{horizontal fragment} $F$ generated from $mp$ is defined as $F=\{V(F),$ $E(F),L^{\prime\prime} \}$, where (1) $V(F)\subseteq V(G)$ is the set of vertices occurring in $\llbracket mp \rrbracket_G$; (2) $E(F) \subseteq E(G)$ is the set of edges occurring in $\llbracket mp \rrbracket_G$; and (3) $L^{\prime\prime} \subseteq L$ is the set of edge labels occurring in $\llbracket mp \rrbracket_G$.

Then, given a set of structural minterm predicates ${M}=\{mp_1,mp_2,$ $...,mp_y\}$, the corresponding \emph{horizontal fragmentation} is $\mathcal{F}=\{F_i|0\le i \le y$ and $F_i$ is the horizontal fragment generated from $mp_i$.$\}$
\end{definition}

\begin{example}
Given the structural simple predicates in Example \ref{example:SimplePredicates}, we can get all structural minterm predicates from frequent access pattern $p_3$ as follows:
(1). $mp_1:p_3(?x0)=Aristotle\land  p_3(?x1)=Ethics$;
(2) $mp_2:p_3(?x0)= Aristotle\land  p_3(?x1)\ne Ethics$;
(3). $mp_3:p_3(?x0)\ne Aristotle\land  p_3(?x1)=Ethics$;
(4). $mp_4:p_3(?x0)\ne Aristotle\land  p_3(?x1)\ne Ethics$.

Figure \ref{fig:ExamplePrimaryHorizontalFragments} shows all horizontal fragments generated from the above structural minterm predicates.
\end{example}

\begin{figure}
       \subfigure[][{ Example Horizontal Fragment Generated from $mp_1$}]{%
      \includegraphics[scale=0.31]{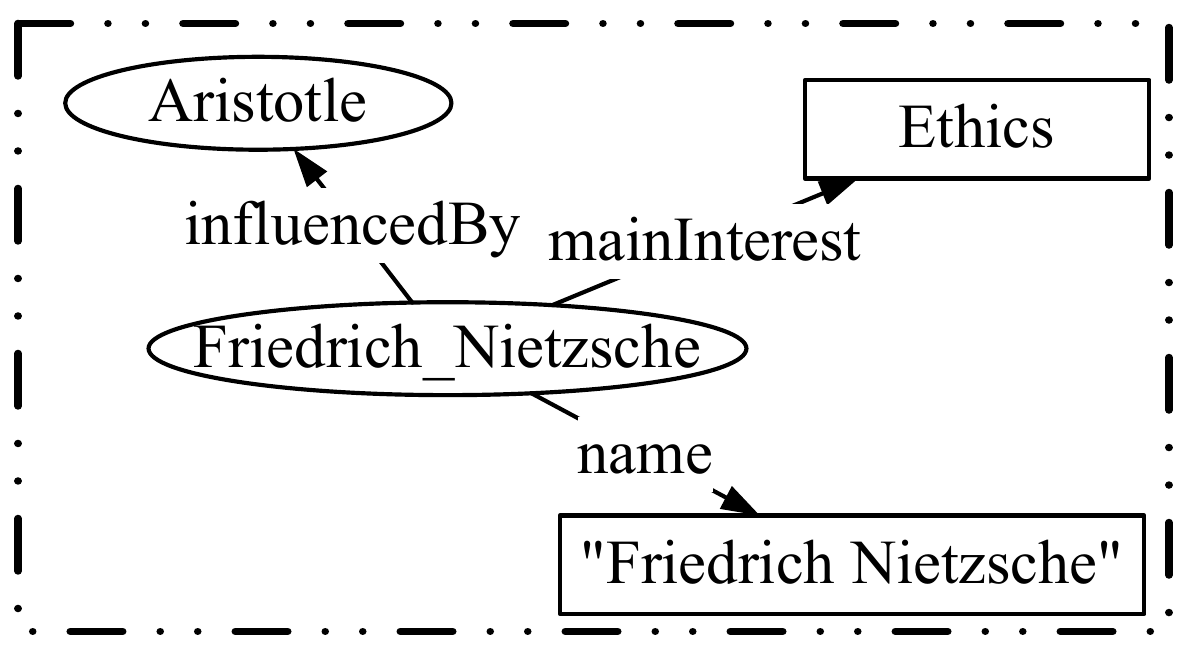}
       \label{fig:m1}%
       }
       \hspace{0.1in}
   \subfigure[][{ Example Horizontal Fragment Generated from $mp_2$}]{%
      \includegraphics[scale=0.31]{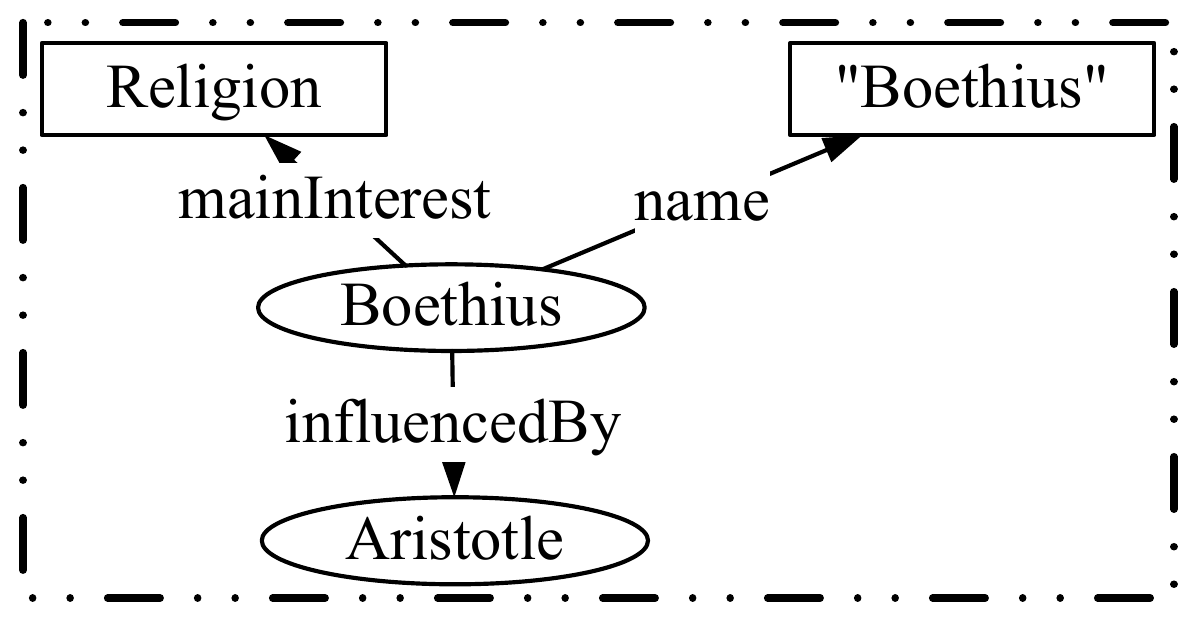}
       \label{fig:m2}%
       }%
       \\
   \subfigure[][{ Example Horizontal Fragment Generated from $mp_3$}]{%
      \includegraphics[scale=0.31]{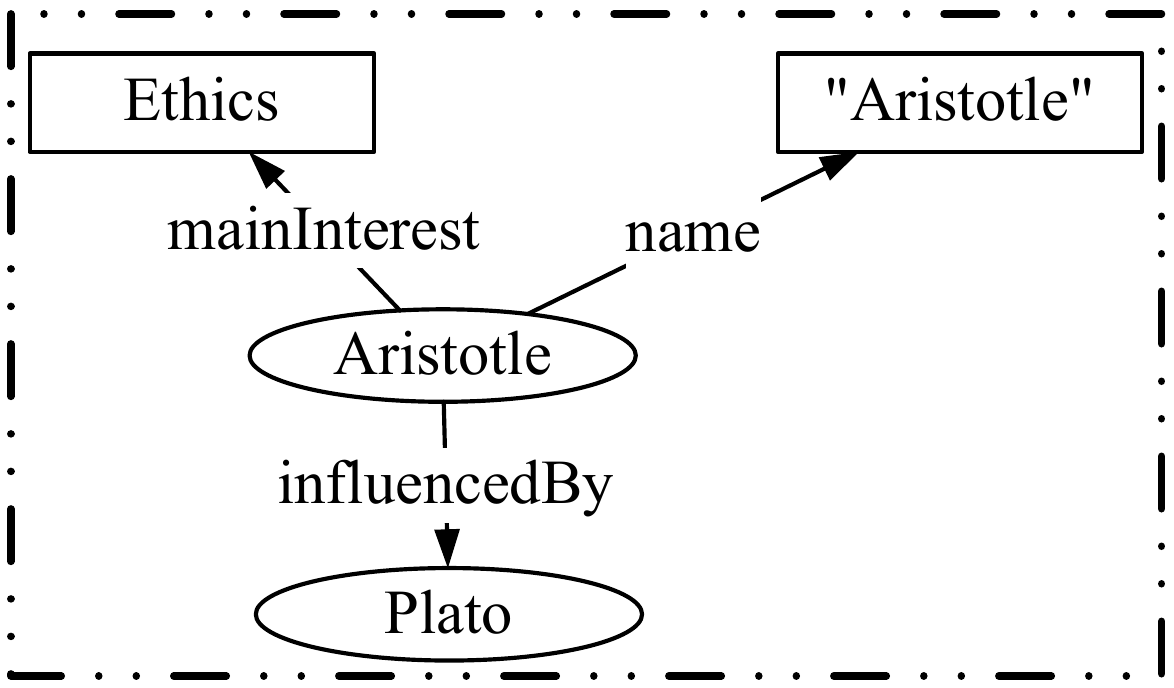}
       \label{fig:m3}%
       }%
   \subfigure[][{Example Horizontal Fragment Generated from $mp_4$}]{%
      \includegraphics[scale=0.31]{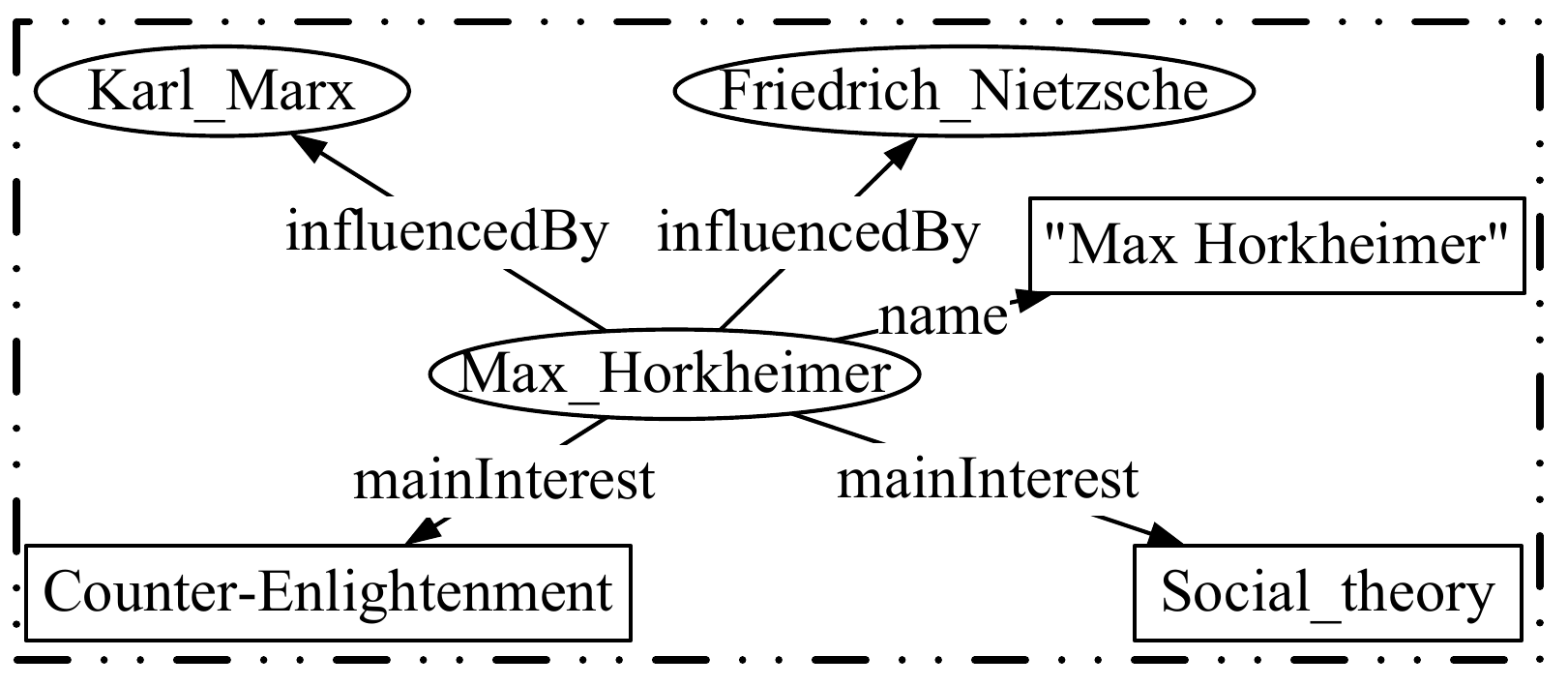}
       \label{fig:m4}%
       }
       \vspace{-0.2in}
\caption{ Example Horizontal Fragments}%
 \label{fig:ExamplePrimaryHorizontalFragments}
\end{figure}

\vspace{-0.15in}
\section{Allocation}\label{sec:Allocation}
After fragmenting the RDF graph, the next step is to allocate all fragments on several sites. In real applications, some frequent access patterns or structural minterm predicates are usually accessed together, so their corresponding fragments should be placed in one site to further avoid the cross-fragments joins. There is a need for some measures evaluating precisely the notion of ``togetherness''. This measure is the affinity of fragments, which indicates how closely related the fragments are.

We define \emph{fragment affinity metric} to measure the togetherness between two fragments generated from frequent access patterns or structural minterm predicates as follows:

\begin{definition}\label{def:applicationpropertyaffinitymeasure}
\textbf{(Fragment Affinity Metric)}
The \emph{fragment affinity metric} between two fragments $F$ and $F^\prime$ with respect to the workload $\mathcal{Q} = \{Q_1,Q_2,...,$ $Q_q\}$ is defined as follows
\begin{itemize}
  \item $aff({F},{F^\prime}) = \sum_{k=1}^{q} {use({Q_k},{p})} \times {use({Q_k},{p^\prime})}$, if $F$ and $F^\prime$ are vertical fragments generated from frequent access patterns $p$ and $p^\prime$;
  \item $aff({F},{F^\prime}) = \sum_{k=1}^{q} {use({Q_k},{mp})} \times {use({Q_k},{mp^\prime})}$, if $F$ and $F^\prime$ are horizontal fragments generated from structural minterm predicates $mp$ and $mp^\prime$;
\end{itemize}
\end{definition}

Based on the fragment affinity metric, we can show how closely related the fragments are. If the affinity metric of two fragments is large, it means that these two fragments are often involved by the same query. Some fragments are so related that they should be placed together to reduce the number of cross-sites joins. Here, we group all fragments into some clusters. The result of clustering corresponds to an allocation $\mathcal{A}$, and each cluster corresponds to an element of $\mathcal{A}$, which means that all fragments in the cluster are placed into the same site.

There are many clustering algorithms to cluster all fragments and we need to select one of them. In this paper, we extend a graph clustering algorithm, PNN \cite{DBLP:PNN}, to cluster all fragments into an allocation $\mathcal{A}=\{A_1,A_2,...,A_m\}$. All fragments of the same cluster are put into one site.

First, we build the \emph{allocation graph} as follows.

\begin{definition}(\textbf{Allocation Graph})\label{def:AllocationGraph} Given a fragmentation $\mathcal{F}=\{F_1,F_2,...,F_n\}$, the corresponding \emph{allocation graph} $AG=\{V(AG),$ $E(AG),f_W\}$ is defined as follows:
\begin{itemize}
  \item $V(AG)$ is a set of vertices that map to all fragments;
  \item $E(AG)$ is a set of undirected edges that $\overline {{v}{v^\prime}} \in E(VG)$ if and only if the fragment affinity metric between the corresponding fragments of $v$ and $v^\prime$ is larger than 0;
  \item $f_W$ is a weight function $f_W:E(AG)\to N^+$. If $v$ and $v^\prime$ correspond to fragments $F$ and $F^\prime$, $f_W(\overline {{v}{v^\prime}}) = aff({F},{F^\prime})$.
\end{itemize}
\end{definition}

Then, the allocation problem is equivalent to cluster all fragments in $m$ clusters, and all fragments in a cluster are connected in $AG$. We define the \emph{density} of a cluster $A_i$ in $AG$ to rate the quality of $A_i$ as follows.
\[\delta (A_i) = \frac{\sum\limits_{{v_i}\in A_i \land {v_j}\in A_i \land \overline {{v_i}{v_j}} \in E(AG)} {f_W(\overline {{v_i}{v_j}})}}{
 \left( {\begin{array}{c}
   |A_{i}|  \\
   2 \\
\end{array}} \right)
}   \]
where $\sum\limits_{{v_i}\in A_i \land {v_j}\in A_i \land \overline {{v_i}{v_j}} \in E(AG)} {f_W(\overline {{v_i}{v_j}})}$ is the sum of weights of all edges in $A_i$ and $ \left( {\begin{array}{c}
   |A_{i}|  \\
   2 \\
\end{array}} \right)$ is the maximum possible number of edges.

The objective of our allocation algorithm is to search for $m$ subgraphs of $AG$ that have the highest densities. Unfortunately, this problem is NP-complete \cite{DBLP:GraphClusteringNPC}, so we propose a heuristic solution as Algorithm \ref{alg:AllocationAlgorithm}. Algorithm \ref{alg:AllocationAlgorithm} is a variant of PNN and picks the locally optimal choice of merging two vertices in $AG$ at each step. Because our objective function can guarantee the locally optimal choice is also the optimal choice for the overall solution, Algorithm \ref{alg:AllocationAlgorithm} can find out the optimal clustering result of $AG$.

Generally speaking, we initialize a cluster for each fragment. Then, we repeatedly picks the two clusters (singletons or larger) that have the highest weight value to be merged. The weight between two clusters are the density value of merging them. Such merging is iterated until the size of the allocation graph has been reduced to $m$.

\vspace{-0.1in}
\begin{algorithm}[h] \label{alg:AllocationAlgorithm}
\caption{Allocation Algorithm}
\small
\KwIn{The allocation graph $AG$ and the preset threshold $\theta$}
\KwOut{ An allocation $\mathcal{A}=\{A_1,A_2,...,A_m\}$}

\For{each vertex $v_i$ in $V(VG)$}
{
    $A_i\gets \{v_i\}$;\\
}
Find the edge $e_{max}$ with the highest weight in $E(AG)$;\\
Initialize $AG^\prime$ that is the same to $AG$;\\
\While{$|V(AG^\prime)|\ne m$}
{
    Generating $AG^\prime$ from $AG$ by merging $e_{max}=\overline {{A_i}{A_{j}}}$ to $A_{ij}$;\\
    \For{each $A_k$ adjacent to ${A_{ij}}$ in $E(AG^\prime)$}
    {
        $f_W(\overline {{A_k}{A_{ij}}})\gets \frac{\sum\limits_{{v_i}\in {A_k} \land ({v_j}\in {A_i}\lor {v_j}\in {A_j}) \land \overline {{v_i}{v_j}} \in E(AG)} {f_W(\overline {{v_i}{v_j}})}}{{\left( \begin{array}{l}
 |{A_k^\prime}| \\
 2 \\
 \end{array} \right)}}  $
    }
    Find the edge $e_{max}$ with the highest weight in $E(AG^\prime)$;\\
}

\end{algorithm}

\vspace{-0.25in}
\section{Distributed Query Processing}\label{sec:QueryProcessing}

\begin{figure*}
       \subfigure[][{A New Input Query $Q_4$}]{%
      \includegraphics[scale=0.35]{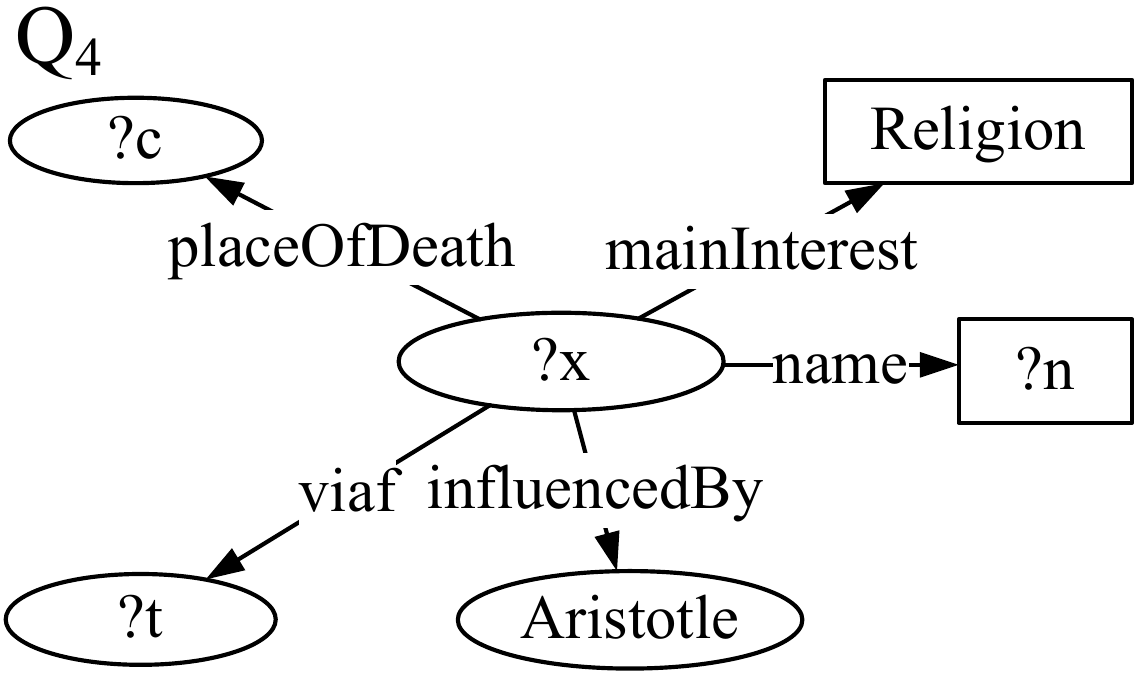}
       \label{fig:ExampleAdditionalQuery}%
       }
       \hspace{0.1in}
   \subfigure[][{Valid Decomposition $\mathcal{D}_1$}]{%
      \includegraphics[scale=0.35]{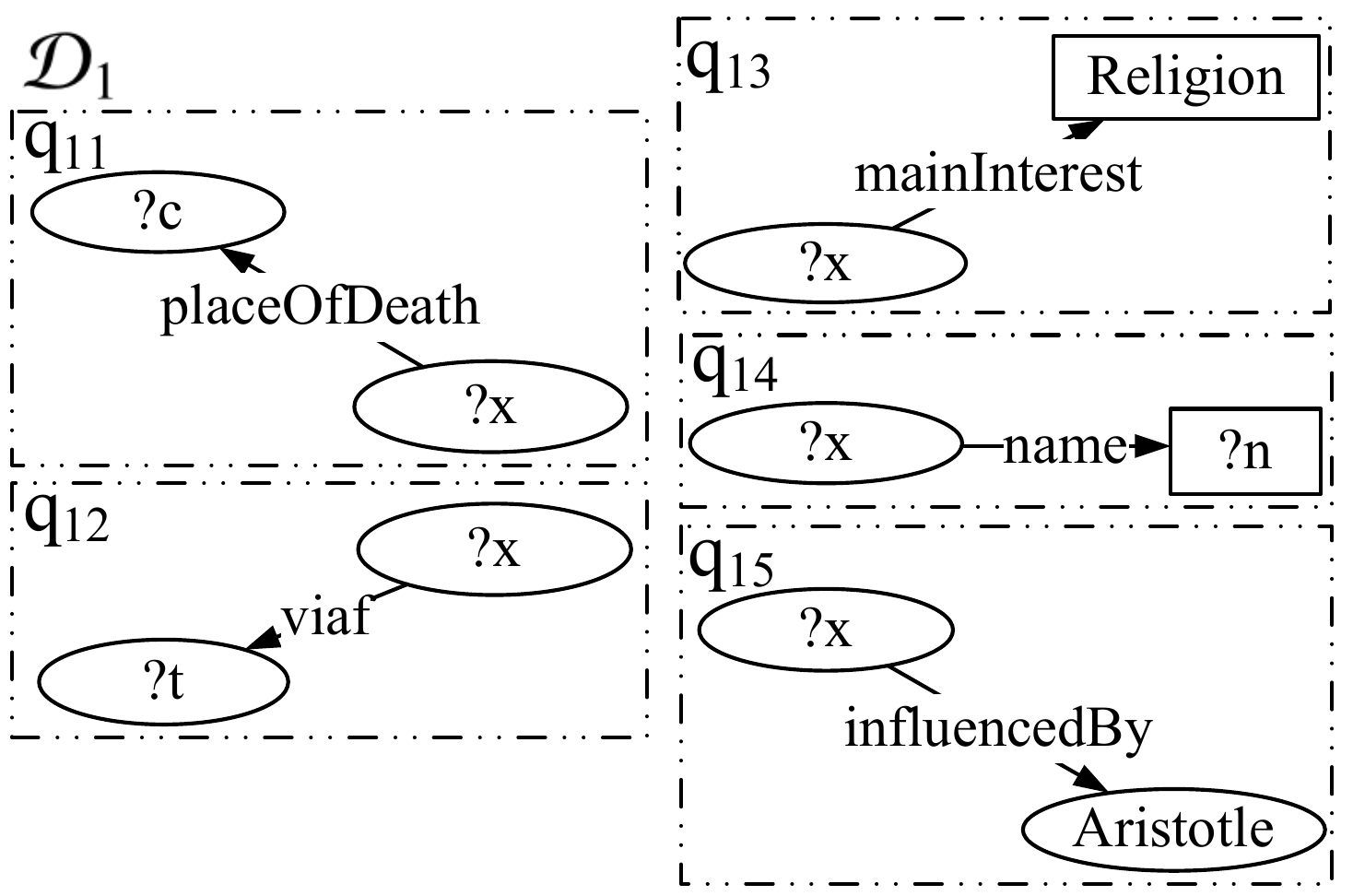}
       \label{fig:QueryDecomposition1}%
       }
       \hspace{0.1in}
   \subfigure[][{Valid Decomposition $\mathcal{D}_2$}]{%
      \includegraphics[scale=0.35]{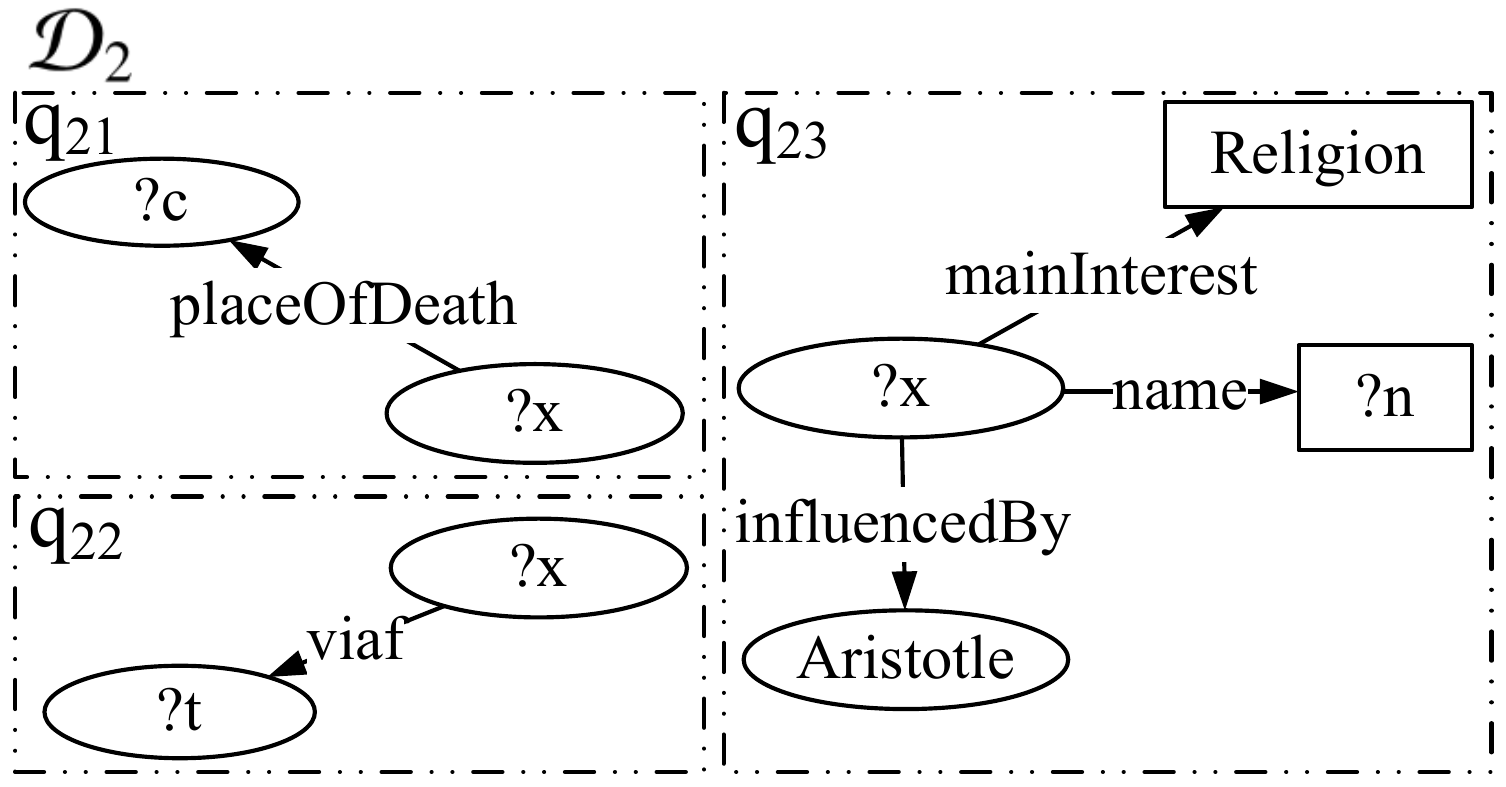}
       \label{fig:QueryDecomposition2}%
       }
       \vspace{-0.2in}
\caption{A New Input Query and Its Example Valid Decompositions}%
 \label{fig:ExampleDecomposition}
\end{figure*}

In this section, we discuss how to process a SPARQL query. For query processing, the metadata is necessary and we introduce how to maintain the metadata in a data dictionary in Section \ref{sec:DataDictionary}. Then, we discuss how to decompose a query into some subqueries in Section \ref{sec:QueryDecomposition}. Last, we discuss how to produce a distributed execution plan and execute all subqueries based on the plan in Section \ref{sec:QueryOptimization}.

\subsection{Data Dictionary}\label{sec:DataDictionary}
After fragmentation and allocation, the results of fragmentation and allocation need to be stored and maintained by the system. This information is necessary during distributed query processing. This information is stored in a data dictionary. The data dictionary stores a global statistics file generated at fragmentation and allocation time. It contains the following information: fragment definitions, their sizes, site mappings, access frequencies and so on.

Since each fragment corresponds to a frequent access pattern or a structural minterm predicate, the data dictionary uses the frequent access pattern with/without constraints as the representative of a fragment. Each frequent access pattern with/without constraints corresponds to a fragment and is associated with all statistics of the fragment. The data dictionary need to fast retrieve all frequent access patterns with/without constraints to determine the relevant frequent access pattern for a query.

We build a hash table to achieve the above objective. We first use the DFS coding \cite{DBLP:gIndex} to translates frequent access patterns into sequences. With the DFS code of a frequent access pattern, we can map any frequent access pattern to an integer by hashing its canonical label. Then, we use the hash table to locate frequent access patterns and retrieve the statistics of their corresponding fragments.

\subsection{Query Decomposition}\label{sec:QueryDecomposition}
When users input a query $Q$, the system first uses the data dictionary to determine which fragments are involved in the query and decomposes the query into some subqueries on fragments.

Given a query $Q$, a decomposition of $Q$ is a set of subqueries $\mathcal{D}=\{q_1,q_2,...,q_t\}$ such that (1) each $q_i$ is a subgraph of $Q$ and $q_i$ maps to a frequent access pattern or structural minterm predicate; (2) $V(q_1)  \cup ... \cup V(q_t)  = V(Q)$; and (3) $E(q_1)  \cup ... \cup E(q_t)  = E(Q)\land \forall i\ne j, E(q_i)\cap E(q_j)=\emptyset$.

Since we partition the RDF graph based on the frequent access patterns, we also decompose the query based on the frequent access patterns. In other words, we decompose the query into subqueries that are homomorphic to frequent access patterns. If a query involves infrequent properties that cannot be decomposed into subqueries homomorphic to any frequent access patterns, then each connected subgraph of the query that only contains infrequent properties corresponds to a subquery. We define the \emph{valid} decomposition as follows.

\vspace{-0.0in}
\begin{definition}(\textbf{Valid Decomposition})
Given a SPARQL query $Q$, a \emph{valid} decomposition $\mathcal{D}=\{q_1,q_2,...,q_t\}$ of $Q$ should meet the following constraint: if $q_i$ ($1\le i \le t$) is not homomorphic to any frequent access patterns, all edges in $q_i$ should be cold edges.
\end{definition}
\vspace{-0.0in}

There exist at least one valid decompositions. A possible decomposition is the decomposition of all subqueries of a single edge. Because we select all frequent access patterns of one edge, the decomposition of all subqueries of a single edge is valid.
Besides the valid decomposition, there may also exist some other valid decompositions. Hence, we propose a cost-model driven selection and the best valid decomposition is the valid decomposition with the smallest cost.

Here, we assume that the cost of a decomposition is the cost of joining all matches of the subqueries in $\mathcal{D}$ and each pair of subqueries' matches can join together. The assumption is the worst case, so that we can quantify the worst-case performance. Then, we define the cost of a decomposition as follows.
\[cost(\mathcal{D}) = \prod\limits_{{q_i} \in \mathcal{D}} {card({q_i})} \]
where $card(q_i)$ is the number of matches for $q_i$, which can be estimated by looking up the data dictionary.

\begin{example}
Assume that an user inputs a new query $Q_4$ as shown in Figure \ref{fig:ExampleAdditionalQuery}. Given frequent access patterns in Figure \ref{fig:generalizedSPARQL}, there can be two valid decompositions $\mathcal{D}_1$ and $\mathcal{D}_2$ as shown in Figures \ref{fig:QueryDecomposition1} and \ref{fig:QueryDecomposition2}. For vertical fragmentation, $q_{23}$ in $\mathcal{D}_2$ is evaluated on the vertical fragment of $p_3$ (Figure \ref{fig:ExampleVerticalFragment}); for horizontal fragmentation, $q_{23}$ is evaluated on the horizontal fragment of $mp_2$ (Figure \ref{fig:m2}).

Whether in vertical or in horizontal fragmentation, it is obvious that $\mathcal{D}_2$ has fewer subqueries than $\mathcal{D}_1$ and $card(q_{23})< card(q_{13})\times card(q_{14}) \times card(q_{15})$. Hence, $cost(\mathcal{D}_2)$ is smaller than $cost(\mathcal{D}_1)$, and $\mathcal{D}_2$ is more of a priority as the final decomposition.
\end{example}

Based on the above definitions, we propose the query decomposition algorithm as Algorithm \ref{alg:QueryDecomposition}. Because the SPARQL query graphs in real applications usually contain $10$ or fewer edges, we can use a brute-force implementation to enumerate all possible decompositions and find the decomposition with the smallest cost.

\vspace{-0.05in}
\begin{algorithm}[h] \label{alg:QueryDecomposition}
\caption{Query Decomposition Algorithm}
\small
\KwIn{A query $Q$}
\KwOut{ A valid decomposition $\mathcal{D}=\{q_1,q_2,...,q_t\}$ of query $Q$}

$MinCost\gets +\infty$;\\
Initialize $\mathcal{D}$ as the decomposition of all subqueries of a single edge;\\
\For{each possible valid decomposition $\mathcal{D}^\prime=\{q_1,...,q_t\}$}
{
    $CurrentCost\gets 1$;\\
    \For{each query $q_i$ in $\mathcal{D}^\prime$}
    {
        Estimate the number of results for $q_i$ as $card(q_i)$ based on the data dictionary;\\
        $CurrentCost\gets CurrentCost \times card(q_i)$\\
    }
    \If{$MinCost>CurrentCost$}
    {
        $\mathcal{D} \gets\mathcal{D}^\prime$;\\
        $MinCost\gets CurrentCost$;\\
    }
}
Return $\mathcal{D}$;
\end{algorithm}
\vspace{-0.05in}

\subsection{Query Optimization and Execution}\label{sec:QueryOptimization}
After decomposing the query, the next step is to find an execution plan for the query which is close to optimal. In this section, we discuss the major optimization issue of finding execution plan, which deals with the join ordering of subqueries. We extend the algorithm of System-R \cite{CiteSeerX:SystemR} to find the optimal execution plan for distributed SPARQL queries. The algorithm is described in Algorithm \ref{alg:QueryOptimization}.

Generally speaking, Algorithm \ref{alg:QueryOptimization} is a variant of System-R style dynamic programming algorithm. It firstly generates the best execution plan of $n-1$ subqueries, and then join the matches of $n-1$ subqueries with the matches of $n$-th subquery. The cost of an execution plan can also be estimated based on the number of subqueries' results, which is stored in the data dictionary.

Finally, each subquery is executed in the corresponding sites in parallel. The optimization of each subquery uses the existing methods in centralized RDF database systems. After the matches of all subqueries are generated, we join them together according to the optimal execution plan.

\begin{algorithm}[h] \label{alg:QueryOptimization}
\caption{Query Optimization Algorithm}
\small
\KwIn{A decomposition $\mathcal{D}=\{q_1,q_2,...,q_t\}$ of query $Q$}
\KwOut{ An execution plan $(...((q_{i1}\Join q_{i2})\Join q_{i3})\Join ... \Join q_{it})$}

\For{each two subqueries $(q_{i})$ and $(q_{j})$ where $1 \le i \ne j \le t$}
{
    Initialize an execution plan $q_{i}\Join q_{j}$ and estimate its cost;\\
    Store all execution plans and their costs in a table $T_2$;\\
}
\For{$i=3$ to $t$}
{
    \For{each execution plan $pl_j$ in $T_{i-1}$}
    {
        \For{each subquery $q_k$ that is not contained by $pl_j$}
        {
            Build execution plan $pl_j \Join q_k$ and estimate its cost;\\
            Store this execution plan and its costs in a table $T_{i}$;\\
        }
        \For{each two plans $pl_j$ and $pl_k$ in $T_{i}$}
        {
            \If{$pl_j$ and $pl_k$ map to the same set of subqueries}
            {
                Eliminate one of $pl_j$ and $pl_k$ that has the larger cost;\\
            }
        }
    }
}
Return the execution plan with the minimum cost;
\end{algorithm}

\vspace{-0.2in}
\section{Experimental Evaluation}\label{sec:Experiment}
We conducted extensive experiments to test the effectiveness of our proposed techniques on a real dataset, DBPedia, and a synthetic dataset, WatDiv. In this section, we report the setting of test data and various performance results.

\subsection{Setting}
\textbf{DBPedia}. DBPedia\footnote{http://km.aifb.kit.edu/projects/btc-2012/dbpedia/} is an RDF dataset extracted from Wikipedia. The DBPedia contains $163,977,110$ triples. We use the DBpedia SPARQL query-log as the workload. This workload contains queries posed to the official DBpedia SPARQL endpoint in 14 days of 2012. After removing some queries that cannot be handled, there are $8,151,238$ queries in the workload.

\textbf{WatDiv}. WatDiv \cite{DBLP:WatDiv} is a benchmark that enable diversified stress testing of RDF data management systems. In WatDiv, instances of the same type can have the different sets of attributes. For testing our methods, we generate five datasets varying sizes from 50 million to 250 million triples. By default, we use the RDF dataset with 100 million triples. In addition, WatDiv can generate a workload by instantiating some templates with actual RDF terms from the dataset. WatDiv provides 20 templates to generate test queries. We use these benchmark templates to generate a workload with 2000 test queries.

We conduct all experiments on a cluster of 10 machines running Linux, each of which has one CPU with four cores of 3.06GHz. Each site has 16GB memory and 150GB disk storage. We select one of these sites as a control site. At each site, we install gStore \cite{DBLP:gStore} to find matches. We use MPICH-3.0.4 running on C++ to join the results generated by subqueries.

For fair performance comparison, we use gStore and MPICH-3.0.4 to re-implement two recent distributed RDF fragmentation strategies. The first one is SHAPE \cite{DBLP:TripleGroup}, which defines a vertex and its neighbors as a triple group and assigns the triple groups according to the value of its center vertices. There are many different kinds of triple groups in \cite{DBLP:TripleGroup} and we use the subject-object-based triple groups in this paper. The second one is WARP \cite{DBLP:WARP}. WARP first uses METIS \cite{DBLP:metis} to divide the RDF graph into fragments. Then, it replicates all matches of a query pattern that cross two fragments in one fragment. We use all frequent access patterns to extend the fragments in WARP.

\subsection{Parameter Setting}\label{sec:minSupImpact}
Our frequent access patterns selection method uses a parameter: $minSup$. In this subsection, we discuss how to set up $minSup$ to optimize query processing. Note that, since the numbers of query templates and queries per query template in WatDiv are specified by users, the parameters can also be determined beforehand. Thus, we only discuss how to set the parameters for DBPedia.

\nop{
\subsubsection{$minSup$}\label{sec:minSupImpact}
}
Given a workload $\mathcal{Q}$, we set the support threshold, $minSup$, to find patterns whose access frequencies are larger than $minSup$. It is clear that the smaller $minSup$ is, the larger number of frequent access patterns there are. More frequent access patterns mean that a query in the workload may have a higher possibility to contain some frequent access patterns.

\vspace{-0.1in}
\begin{figure}[h]%
   \subfigure[$minSup$]{%
		\resizebox{0.5\columnwidth}{!}{
			\input{minSup_FAP}
		}
       \label{fig:FAPNumberDBPedia}%
       }%
\subfigure[Workload Hitting Ratio]{%
		\resizebox{0.5\columnwidth}{!}{
			\input{workload_coverage}
		}
       \label{fig:workloadcoverageDBPedia}%
       }%
\vspace{-0.2in}
 \caption{Effect of Frequent Access Patterns}%
 \label{fig:FAPEffect}
\end{figure}
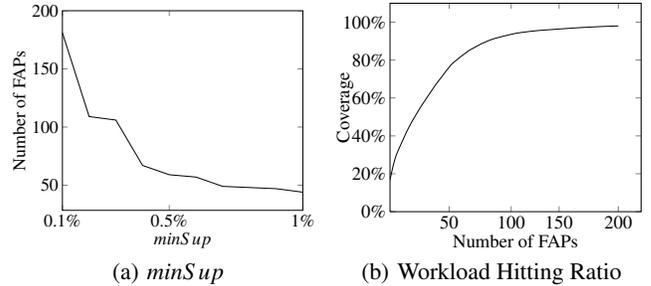

Figure \ref{fig:FAPNumberDBPedia} shows the impact of $minSup$. As $minSup$ increases, the number of frequent access patterns (FAPs) decreases. Hence, when we set $minSup$ as $0.1\%$ of the total number of queries in the workload, there are $163$ frequent access patterns for DBPedia. When $minSup$ is $1\%$ of the total number of queries, the number of frequent access patterns is reduced to $44$ for DBPedia. Furthermore, fewer frequent access patterns means that fewer queries in the workload are hit, as shown in Figure \ref{fig:workloadcoverageDBPedia}.

Even if we set $minSup$ as $0.1\%$ of the total number of queries, the number of frequent access patterns is not large. Hence, in the following, we set $minSup$ as $0.1\%$ of the total number of queries for DBPedia by default.

\subsection{Throughput}
In this experiment, we test the throughput of different fragmentation strategies. We sample $1\%$ of all queries in the workload and measure the throughput in queries per minute. Figure \ref{fig:ThroughputTest} shows the number of queries answered in one minute of different fragmentation strategies.

\begin{figure}[h]
   \subfigure[DBPedia]{%
		\resizebox{0.45\columnwidth}{!}{
			\input{ThroughputTest}
		}
       \label{fig:ThroughputTestDBPedia}%
       }
        \hspace{0.1in}
\subfigure[WatDiv]{%
		\resizebox{0.45\columnwidth}{!}{
			\input{ThroughputTestWatDiv}
		}
       \label{fig:ThroughputTestWatDiv}%
       }%
       \vspace{-0.2in}
 \caption{Throughput Comparison}%
 \label{fig:ThroughputTest}
\end{figure}
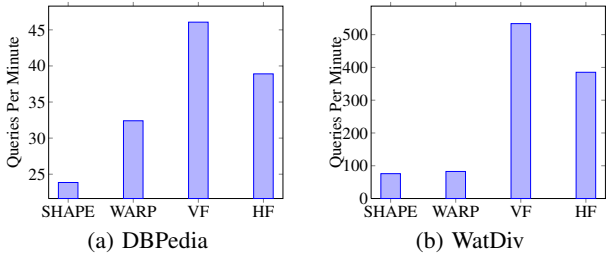

For SHAPE and WARP, each query concerns all fragments, so queries are still processed sequentially. Since WARP is more balanced than SHAPE, the throughput of WARP is a little better than SHAPE. WARP can handle about 32 and 82 queries in one minute for DBPedia and WatDiv, while SHAPE can handle 24 and 75 queries.

For the vertical fragmentation strategy (VF), since a query often only contains a few frequent access patterns, it only involves a few fragments. Two queries involving different fragments can be evaluated in parallel. Hence, about 46 queries and 533 queries can be answered in one minute for DBPedia and WatDiv, respectively. For the horizontal fragmentation strategy (HF), each frequent access pattern specified by the query may map to many structural minterm predicates and the corresponding fragments of these structural minterm predicates may be allocated to different sites. Hence, the throughput of the horizontal fragmentation strategy is a little worse than the vertical fragmentation strategy, and 38 and 385 queries can be answered in one minute for DBPedia and WatDiv.

\subsection{Response Time}
In this experiment, we test the query performance of different fragmentation strategies. We also sample $1\%$ of all queries in the workload and compute the average query response time of a query. Figure \ref{fig:PerformanceTest} shows the performance results.

SHAPE and WARP partition the RDF graph into some subgraphs, and distributes these subgraphs among different sites. The query should be processed in many sites in parallel. Hence, SHAPE is less balanced and sometime need cross-fragment joins, so SHAPE needs about 2.5 and 0.79 seconds to answer a query for DBPedia and WatDiv, while WARP takes 1.8 and 0.72 seconds.

For the vertical fragmentation strategy, only relevant fragments are searched for matches and the search space is reduced. Therefore, a query can be answered in about 0.8 seconds for DBPedia and 0.3 seconds for WatDiv. For the horizontal fragmentation strategy, we can filter out all irrelevant fragments mapping to the structural minterm predicates not specified by the query, which can further reduce the search space. Hence, a query can be answered with about 0.6 seconds for DBPedia and 0.15 seconds for WatDiv.

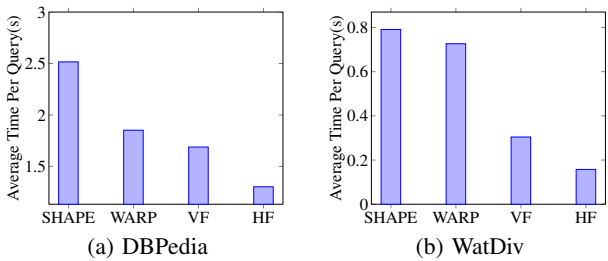
\begin{figure}[h]%
   \subfigure[DBPedia]{%
		\resizebox{0.45\columnwidth}{!}{
			\input{PerformanceTest}
		}
       \label{fig:PerformanceTestDBPedia}%
       }
       \hspace{0.1in}
\subfigure[WatDiv]{%
		\resizebox{0.45\columnwidth}{!}{
			\input{PerformanceTestWatDiv}
		}
       \label{fig:PerformanceTestWatDiv}%
       }%
       \vspace{-0.2in}
 \caption{Performance Comparison}%
 \label{fig:PerformanceTest}
\end{figure}

\vspace{-0.15in}
\subsection{Scalability Test}
In this experiment, we investigate the impact of dataset size on our fragmentation strategies. We generate five WatDiv datasets varying the from 50 million to 250 million triples to test our strategies. Figure \ref{fig:ScalabilityTest} shows the results. Generally speaking, as the size of RDF datasets gets larger, the average response times of one query increase and the numbers of queries answered in one minute decrease accordingly. However, the rates of increase and decrease are slow, and we can say that the query performance and throughput are scalable with RDF graph size on the datasets.

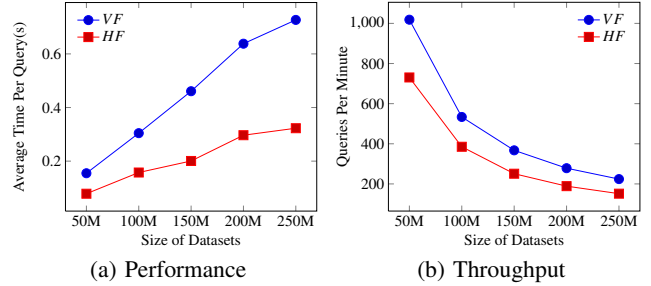
\begin{figure}[h]%
   \subfigure[Performance]{%
		\resizebox{0.5\columnwidth}{!}{
			\input{PerformanceScalability}
		}
       \label{fig:PerformanceScalability}%
       }%
\subfigure[Throughput]{%
		\resizebox{0.5\columnwidth}{!}{
			\input{ThroughputScalability}
		}
       \label{fig:ThroughputScalability}%
       }%
       \vspace{-0.15in}
 \caption{Varying Size of Datasets}%
 \label{fig:ScalabilityTest}
\end{figure}

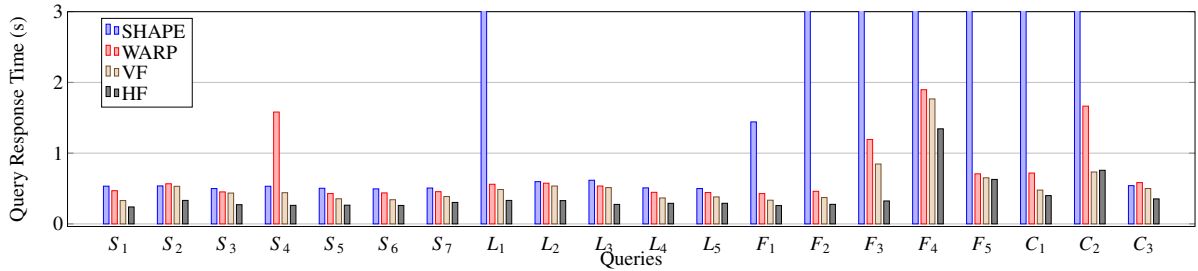
\begin{figure*}%
		\resizebox{0.9\columnwidth}{!}{
			\input{WatDiv100M_comparison}
		}
\vspace{-0.2in}
 \caption{Query Performance of Benchmark Queries}%
 \label{fig:BenchmarkPerformance}
\end{figure*}

\vspace{-0.15in}
\subsection{Redundancy}
Table \ref{table:RedundancyRatio} shows the redundancy ratio of the number of edges in all generated fragments to the total number of edges in the original RDF graph for each fragmentation strategy. For SHAPE, if a fragment contains a vertex with high degree, all adjacent edges of the high degree vertex are introduced. Most of these introduced edges are redundant, and cause the redundancy ratios of SHAPE nearly $3$ for DBPedia and 1.74 for WatDiv. WARP divides the RDF graph while minimizing the edge cut, so the number of edges crossing two fragments for WARP is smaller than the number for SHAPE. Therefore, the redundancy ratio of WARP is smaller. Note that, WatDiv is much denser than DBPedia, so the minimum cut-set for WatDiv contains a higher proportion of edges. Hence, the redundancy ratio of WatDiv is 1.54, but the ratio of DBPedia is only 1.01.

\begin{table}[h]
\caption{Redundancy (Ratio to original dataset)}
\centering
\scriptsize
\begin{tabular}{|c|c|c|}
\hline
 & {DBPedia}& {WatDiv}\\
\hline
  \hline
  SHAPE &2.99 & 1.74\\
  \hline
  WARP & 1.01&1.54\\
  \hline
  VF &1.38 &1.04	\\
  \hline
  HF &1.42 &1.06\\
  \hline
  \end{tabular}
  \vspace{-0.1in}
\label{table:RedundancyRatio}
\end{table}
\vspace{-0.1in}

Our fragmentation strategies find and materialize some frequent access patterns (or structural minterm predicates). As discussed in Section \ref{sec:minSupImpact}, the number of frequent access patterns is limited. Hence, the redundancy ratios of our fragmentation strategies are limited. Note that, the horizontal strategy has a little larger redundancy ratio than the vertical fragmentation strategy. This is because that different structural minterm predicates derived from the same frequent access patterns share some common triple patterns. These common triple patterns may cause more redundant edges.

\subsection{Offline Performance}
Table \ref{table:LoadingPerformance} shows the data partitioning and loading time of the datasets for different fragmentation strategies. Although SHAPE has an almost perfect uniform distribution, its redundancy ratio is too large and each fragment contains too many redundant edges. Hence, loading fragments in SHAPE also takes much time. WARP uses METIS \cite{DBLP:metis}. Since DBPedia is sparse (i.e. $|(E(G)|/|V(G)|\approx 1$), METIS can guarantee that there are a few redundant edges and all fragments have a nearly uniform distribution. Then, WARP has less loading time than SHAPE. However, for WatDiv, the data graph is dense (i.e. $|(E(G)|/|V(G)|\gg 1$), so the fragmentation result of METIS is unbalanced. Then, WARP takes more loading time than SHAPE to load the largest fragments.

Since nearly half of all edges for DBPedia are infrequent edges, loading the cold graph of DBPedia is the bottleneck in our fragmentation strategies. However, in WatDiv, there are not so many infrequent edges. Then, the loading time of our fragmentation strategies for WatDiv is more acceptable. Note that, because the structural minterm predicates are derived from the frequent access patterns, the cold graphs for the vertical and horizontal fragmentation strategies are the same. Thus, the loading times for the vertical and horizontal fragmentation strategies are the same.

\begin{table}[h]
\caption{Partitioning and Loading Time (in min)}
\centering
\scriptsize
\begin{tabular}{|c||c|c|c||c|c|c|}
\hline
 & \multicolumn{3}{c||}{DBPedia}& \multicolumn{3}{c|}{WatDiv}\\
\hline
\hline
 \tabincell{c}{Strategies} & \tabincell{p{1cm}}{Partitioning} & \tabincell{p{0.7cm}}{Loading}& \tabincell{c}{Total}& \tabincell{p{1cm}}{Partitioning} & \tabincell{p{0.7cm}}{Loading}& \tabincell{c}{Total}\\
  \hline
  \nop{
  Cent & -&178	&	178 &	- &19.10	&	19.10\\
  \hline
  Random & 40 &31	&	71 &	&	&	\\
  \hline
  PF & 40 &51	&	91 &	1.86&	3.24&	5.10\\
  \hline
  }
  SHAPE & 41&	30 &	71 & 20	& 19	& 49	\\
  \hline
  WARP &43 &	28&	 71&	33  & 46	&	79 \\
  \hline
 VF  &50 &97	&	147 &	31 &	28 &59	 \\
  \hline
  HF & 58&	97	 &	139 & 34	&  28  &	62 \\
  \hline
  \nop{
  MF & &	97	 &	 &	&    &	\\
  \hline
  }
  \end{tabular}
  \vspace{-0.15in}
\label{table:LoadingPerformance}
\end{table}
\vspace{-0.15in}

\subsection{Experiments for Benchmark Queries}
In this experiment, we compare our methods with other fragmentation strategies on benchmark queries provided by WatDiv. There are 20 benchmark queries in WatDiv, and these queries can be classified into 4 structural categories: linear (L), star (S), snowflake (F) and complex (C). Figure \ref{fig:BenchmarkPerformance} shows the performance of different approaches. Generally speaking, we find out that our methods outperforms other two methods in most cases. This is because that each benchmark query can be decomposed into some frequent access patterns or structural minterm predicates. Hence, our fragmentation strategies can filter out many irrelevant fragments. In contrast, SHAPE and WARP always concern all fragments, and SHAPE further needs some cross-fragment joins for complex queries.

Let us look deeper into Figure \ref{fig:BenchmarkPerformance} and analyze each individual fragmentation strategy. SHAPE has to involve all fragments for any queries, so its performance is always worse than our fragmentation strategies. In particular, for star queries ($S_1$ to $S_7$), the difference between the query response times of SHAPE and our fragmentation strategies is not very large, because the subject-object-based triple groups that we use can guarantee that there is no intermediate result and all star queries can be answered at each fragment locally. However, for other shapes of queries, SHAPE has to decompose the queries and do cross-fragment joins to merge the intermediate results. Then, the performance of SHAPE decreases greatly. Especially for the unselective queries ($L_1$, $F_1$, $F_2$, $F_3$, $F_4$, $F_5$, $C_1$ and $C_2$), the performance of SHAPE is an order of magnitude worse than our fragmentation strategies.

Since WARP also use patterns to replicate triples for avoiding cross-fragment joins in complex queries, WARP has better performance that SHAPE in most case. However, WARP still always concerns all fragments in all sites for any kind of queries. The search space of WARP for a query is higher than our fragmentation strategies. Thus, our fragmentation strategies always result in better performance. Especially for the query of very complex structure ($C_2$), our fragmentation strategies can filter out many irrelevant fragments, which can result in much smaller search space than WARP. Hence, for $C_2$, our strategies is twice as fast as WARP.

Since all benchmark queries are generated from instantiating benchmark templates with actual RDF terms, these benchmark queries always correspond to a limited number of minterm predicates. Hence, the horizontal fragmentation is always faster than the vertical fragmentation.

\vspace{-0.15in}
\section{Related Work}\label{sec:related}
For both the general graph and the RDF graph, as the graph size grows beyond the capability of a single machine, many works \cite{DBLP:Partout,DBLP:WARP,DBLP:GraphPartition,TKDE11:HadoopRDF,DBLP:conf/dasfaa/YangCWCD13,DBLP:TripleGroup,DBLP:conf/IEEEcloud/LeeLTZZ13,SIGMOD14:TriAD,DBLP:conf/icde/WangPS13,DBLP:metis,ICDE13:EAGRE,ICDE14:MLP,ICDE15:PathPartitioning} have been proposed about graph fragmentation and allocation. We can divide all these methods into two categories:
global goal-oriented graph fragmentation methods and local pattern-based graph fragmentation methods.

\textbf{Global Goal-Oriented Graph Fragmentation.}
For this kind of methods \cite{DBLP:metis,DBLP:GraphPartition,ICDE13:EAGRE,ICDE14:MLP,DBLP:journals/pvldb/MargoS15}, they divide $G$ into several fragments while maximizing some goal function. They first transform a large graph into a small graph; then, apply some graph partitioning algorithms on the small graph; finally, the partitions on the small graph are projected back to the original graph. These methods often apply some existing methods (such as KL \cite{BibSonomy:KL}) directly on the transformed graph in the second step. If we track the transforming step, the partitions on the small graph can be easily projected back to the original graphs in the third step. Hence, the largest difference among different graph coarsening-based methods is how to coarsen the original graph into a small graph.

In particular, METIS \cite{DBLP:metis} uses the maximal matching to coarsen the graph. A matching of a graph is a set of edges that no two edges share an endpoint. A maximal matching of a graph is a matching to which no more edges can be added and remain a matching. GraphPartition \cite{DBLP:GraphPartition} directly uses METIS in the RDF graph. WARP \cite{DBLP:WARP} uses some frequent structures in workload to further extend the results of GraphPartition. EAGRE \cite{ICDE13:EAGRE} coarsens the RDF graph by using the entity concept in RDF data. It considers an entity to be a subject and its complete description. By grouping the entities of the same class, an RDF graph can be compressed as a compressed RDF entity graph. MLP \cite{ICDE14:MLP} designs a method to coarsen the graph by label propagation. Vertices with the same label after the label propagation are coarsened to a vertex in the coarsened graph. Sheep \cite{DBLP:journals/pvldb/MargoS15} transform the graph into a elimination tree via a distributed map-reduce operation, and then partition this tree while reducing communication volume. Tomaszuk et. al. \cite{DBLP:conf/bdas/TomaszukSW15} briefly survey how to apply existing graph fragmentaion solutions from the theory of graphs to RDF graphs.

Global goal-oriented graph fragmentation methods assume that if there are few edges crossing different fragments, the communication cost is little. If an application involves nearly all vertices in the graph, few cross-fragments edges indeed result in little communication. A typical application suitable for graph coarsening-based methods is PageRank.

In some applications, one static fragmentation cannot fit all. Hence, Sedge \cite{DBLP:conf/sigmod/YangYZK12} maintains many fragmentations with different crossing edges, while Shang et. al. \cite{DBLP:conf/icde/ShangY13} move some vertices of one fragment to another fragment during graph computing according to the workload. Yan et. al. \cite{DBLP:conf/icde/YanWZQMP09} propose a indexing scheme based on fragmentation to help query engine fast locate the instances.

\textbf{Local Pattern-based Graph Fragmentation.}
For this kind of methods \cite{TKDE11:HadoopRDF,DBLP:conf/dasfaa/YangCWCD13,DBLP:TripleGroup,DBLP:conf/IEEEcloud/LeeLTZZ13,SIGMOD14:TriAD,DBLP:conf/icde/WangPS13,ICDE15:PathPartitioning} , they first find certain patterns as the fragmentation units to cover the whole graph; then, they distribute these patterns into sites. The local pattern-based methods mainly differ in their definitions of the fragmentation unit.

HadoopRDF \cite{TKDE11:HadoopRDF}\nop{ and P-Partition \cite{DBLP:PredicateJoin}} groups triples with the same property together and each group corresponds to a fragmentation unit. Then, they store all fragmentation units over HDFS. Yang et. al.\cite{DBLP:conf/dasfaa/YangCWCD13} define some special query patterns, and subgraphs of a pattern are considered as a fragmentation unit. Lee et. al. \cite{DBLP:TripleGroup,DBLP:conf/IEEEcloud/LeeLTZZ13} define the fragmentation unit as a vertex and its neighbors, which they call a triple group. The triple groups are distributed based on some heuristic rules.  For each vertex, SketchCluster \cite{DBLP:conf/icde/WangPS13} identifies the set of labeled vertices reachable within its one-hop neighborhood as its features and employs the KModes algorithm to group related vertices based on the features. Partout \cite{DBLP:Partout} extends the concepts of minterm predicates in relational database systems, and uses the results of minterm predicates as the fragmentation units. TriAD \cite{SIGMOD14:TriAD} uses METIS \cite{DBLP:metis} to divide the RDF graph into many partitions. Then, each result partition is considered as a unit and distributed among different sites based on a hash function. PathPartitioning \cite{ICDE15:PathPartitioning} uses paths in RDF graphs as fragmentation units.

Local pattern-based graph fragmentation methods assume that some real applications only concerns a part of the whole graph. If an application only concerns the vertices of some certain patterns, these methods only access the relevant fragments and reduce the communication cost across fragments. A typical example application is subgraph homomorphism checking.

\vspace{-0.15in}
\section{Conclusion}\label{sec:Conclusion}
In this paper, we discuss how to manage the large RDF graph in a distributed environment. First, we mine and select some frequent access patterns to partition the RDF graph into many smaller fragments. Then, we propose an allocation algorithm to distribute all fragments over different sites. Last, we discuss how process the query based on the results of fragmentation and allocation. Extensive experiments verify our approaches.

\small{
\textbf{Acknowledgement.} This was supported by 863 project under Grant No. 2015AA015402, NSFC under Grant No. 61532010, 61370055 and 61272344. Lei Chen's work is supported in part by the Hong Kong
RGC Project N HKUST637/13, National Grand Fundamental Research 973
Program of China under Grant 2014CB340303 , NSFC Grant No. 61328202,
NSFC Guang Dong Grant No. U1301253, Microsoft Research Asia Gift
Grant and Google Faculty Award 2013
}

\small
\bibliographystyle{abbrv}
\bibliography{IEEEexample}

\end{document}

%% file: minSup_FAP.tex
\begin{tikzpicture}[font=\Large]

        \begin{axis}[
        xlabel=$minSup$,
        ylabel=Number of FAPs,
        ymax=200,
        xmax=10,
        xmin=1,
        xtick = {1,5,10},
        xticklabels={$0.1\%$,$0.5\%$,$1\%$},
    ]
      \addplot[ black] plot coordinates {
        (1,     181)
        (2,     109)
        (3,     106)
        (4,     67)
        (5,    59)
        (6,    57)
        (7,    49)
        (8,    48)
        (9,    47)
        (10,    44)
    };

    \end{axis}

\end{tikzpicture}

%% file: workload_coverage.tex
\begin{tikzpicture}[font=\Large]

        \begin{axis}[
        xlabel=Number of FAPs,
        ylabel=Coverage,
        ymax=110,
        ymin=0,
        xmin=1,
        xtick = {22,44,61,82},
        xticklabels={$50$,$100$,$150$,$200$},
        ytick = {0,20,40,60,80,100},
        yticklabels={$0\%$,$20\%$,$40\%$,$60\%$,$80\%$,$100\%$},
    ]
      \addplot[smooth, black] plot coordinates {
     (1, 16.72)
(2, 23.54)
(3, 28.90)
(4, 32.62)
(5, 35.95)
(6, 39.24)
(7, 42.48)
(8, 45.35)
(9, 48.08)
(10, 50.60)
(11, 53.09)
(12, 55.49)
(13, 57.75)
(14, 59.98)
(15, 62.19)
(16, 64.38)
(17, 66.56)
(18, 68.55)
(19, 70.50)
(20, 72.45)
(21, 74.40)
(22, 76.24)
(23, 78.05)
(24, 79.28)
(25, 80.49)
(26, 81.71)
(27, 82.82)
(28, 83.92)
(29, 85.02)
(30, 85.85)
(31, 86.67)
(32, 87.48)
(33, 88.29)
(34, 88.99)
(35, 89.64)
(36, 90.28)
(37, 90.82)
(38, 91.34)
(39, 91.73)
(40, 92.13)
(41, 92.53)
(42, 92.89)
(43, 93.26)
(44, 93.58)
(45, 93.91)
(46, 94.21)
(47, 94.43)
(48, 94.64)
(49, 94.83)
(50, 95.02)
(51, 95.20)
(52, 95.35)
(53, 95.51)
(54, 95.65)
(55, 95.77)
(56, 95.87)
(57, 95.98)
(58, 96.08)
(59, 96.19)
(60, 96.29)
(61, 96.40)
(62, 96.50)
(63, 96.60)
(64, 96.71)
(65, 96.81)
(66, 96.92)
(67, 97.01)
(68, 97.10)
(69, 97.20)
(70, 97.27)
(71, 97.35)
(72, 97.42)
(73, 97.50)
(74, 97.57)
(75, 97.64)
(76, 97.70)
(77, 97.77)
(78, 97.83)
(79, 97.88)
(80, 97.93)
(81, 97.98)
(82, 98.03)
    };

    \end{axis}

\end{tikzpicture}

%% file: ThroughputTest.tex
\begin{tikzpicture}[font=\LARGE]
\begin{axis}[
ybar,
width = 10cm,
height = 8.5cm,
  xtick = {4,5,6,7},
   xticklabels={SHAPE, WARP, VF, HF},
   ylabel={Queries Per Minute},
   bar width=20pt,
  ]
   \addplot plot coordinates {
        (4,    23.85152426)

        (5,    32.40359323)
        (6,    46.06632899)
        (7,     38.90192824)
    };
\end{axis}

\end{tikzpicture}

%% file: ThroughputTestWatDiv.tex
\begin{tikzpicture}[font=\LARGE]
\begin{axis}[
ybar,
width = 10cm,
height = 8.5cm,
ymin =0,
  xtick = {4,5,6,7},
   xticklabels={SHAPE, WARP, VF, HF},
   ylabel={Queries Per Minute},
   bar width=20pt,
  ]
   \addplot plot coordinates {
        (4,    75.87026354)

        (5,    82.59229689)
        (6,    533.6685562)
        (7,    385.18270507)
    };
\end{axis}

\end{tikzpicture}

%% file: PerformanceTest.tex
\begin{tikzpicture}[font=\LARGE]
\begin{axis}[
ybar,
width = 10cm,
height = 8.5cm,
ymax =3,
  xtick = {2,3,4,5,6,7},
   xticklabels={Random, PF, SHAPE, WARP, VF, HF},
   ylabel={Average Time Per Query(s)},
   bar width=20pt,
  ]
   \addplot plot coordinates {

        (2,     19.61643)
        (3,     15.065685)
        (4,     2.5155625)

        (5,    1.8516465)
        (6,    1.688238)
        (7,    1.302155)
    };
\end{axis}

\end{tikzpicture}

%% file: PerformanceTestWatDiv.tex
\begin{tikzpicture}[font=\LARGE]
\begin{axis}[
ybar,
width = 10cm,
height = 8.5cm,
ymin =0,
  xtick = {4,5,6,7},
   xticklabels={SHAPE, WARP, VF, HF},
   ylabel={Average Time Per Query(s)},
   bar width=20pt,
  ]
   \addplot plot coordinates {

        (4,     0.790823667 )

        (5,    0.72646)
        (6,    0.304368333)
        (7,    0.1575463)
    };
\end{axis}

\end{tikzpicture}

%% file: PerformanceScalability.tex
\begin{tikzpicture}[font=\large]
    \begin{axis}[
        xlabel=Size of Datasets,
        ylabel=Average Time Per Query(s),
        xtick = {10,30,50,70,90},
        xticklabels={50M,100M,150M,200M,250M},
        legend cell align=left,
        legend style={draw=none},
        legend pos= north west
    ]
    \addplot plot[mark size=3.5pt]  coordinates {
        (10,    0.154831019)
        (30,    0.304368333)
        (50,    0.460815078)
        (70,    0.63807107)
        (90,   0.726799136)
    };
    \addplot plot[mark size=3.5pt]  coordinates {
        (10,    0.07847678)
        (30,    0.1575463)
        (50,    0.200629654)
        (70,    0.296674878)
        (90,   0.322782528)
    };
    \legend{$VF$\\$HF$\\}

    \end{axis}
\end{tikzpicture}

%% file: ThroughputScalability.tex
\begin{tikzpicture}[font=\large]
    \begin{axis}[
        xlabel=Size of Datasets,
        ylabel=Queries Per Minute,
        xtick = {10,30,50,70,90},
        xticklabels={50M,100M,150M,200M,250M},
        legend cell align=left,
        legend style={draw=none},
        legend pos= north east
    ]
    \addplot plot[mark size=3.5pt]  coordinates {
        (10,    1018.380415)
        (30,    533.6685562)
        (50,    367.368919)
        (70,    278.387651)
        (90,   224.1063406)
    };
    \addplot plot[mark size=3.5pt]  coordinates {
        (10,    730.5117913)
        (30,    385.18270507)
        (50,    250.8636195)
        (70,    188.8612349)
        (90,   151.4335606)
    };
    \legend{$VF$\\$HF$\\}

    \end{axis}
\end{tikzpicture}

%% file: WatDiv100M_comparison.tex
	\begin{tikzpicture}[font=\Large]
 		 \begin{axis}[
                anchor={(0,100)},
               width = 30cm,
               height = 7cm,
    			major x tick style = transparent,
    			ybar,
   	ymajorgrids = true,
   ymax=3,
   			ylabel = {Query Response Time (s)},
    			xlabel = {Queries},
    			symbolic x coords = {$S_1$,$S_2$,$S_3$,$S_4$,$S_5$,$S_6$,$S_7$,$L_1$,$L_2$,$L_3$,$L_4$,$L_5$,$F_1$,$F_2$,$F_3$,$F_4$,$F_5$,$C_1$,$C_2$,$C_3$},
    			scaled y ticks = true,
			bar width=4pt,
             enlarge x limits=0.05,
			legend pos= north west,
 legend cell align=left
   		]

    		\addplot coordinates {($S_1$, 0.533234) ($S_2$, 0.53676) ($S_3$, 0.499481) ($S_4$, 0.53192) ($S_5$, 0.503804) ($S_6$, 0.494941) ($S_7$, 0.507011
) ($L_1$, 8.16512) ($L_2$, 0.596883) ($L_3$, 0.617002) ($L_4$, 0.50908) ($L_5$, 0.499741) ($F_1$, 1.44148) ($F_2$, 5.45851) ($F_3$, 8.50395) ($F_4$, 4.84184) ($F_5$, 15.4971) ($C_1$, 15.0813) ($C_2$, 10) ($C_3$, 0.540517)};

    \addplot coordinates {($S_1$, 0.469474) ($S_2$, 0.567474) ($S_3$, 0.45226) ($S_4$, 1.58119) ($S_5$, 0.428684) ($S_6$, 0.438109) ($S_7$, 0.45545) ($L_1$, 0.560161) ($L_2$, 0.575275) ($L_3$, 0.535403) ($L_4$, 0.446221) ($L_5$, 0.443223) ($F_1$, 0.428258) ($F_2$, 0.460722) ($F_3$, 1.19331) ($F_4$, 1.8969) ($F_5$, 0.70799) ($C_1$, 0.718769) ($C_2$, 1.66483) ($C_3$, 0.583648)};

    \addplot coordinates {($S_1$, 0.33) ($S_2$, 0.531087) ($S_3$, 0.436131) ($S_4$, 0.440956) ($S_5$, 0.355201) ($S_6$, 0.341694) ($S_7$, 0.386994) ($L_1$, 0.485984) ($L_2$, 0.535543) ($L_3$, 0.513117) ($L_4$, 0.366483) ($L_5$, 0.3813) ($F_1$, 0.334951) ($F_2$, 0.372581) ($F_3$, 0.84638) ($F_4$, 1.76639) ($F_5$, 0.6513025) ($C_1$, 0.477901) ($C_2$, 0.733719) ($C_3$, 0.500482)};

    \addplot coordinates {($S_1$, 0.2395) ($S_2$, 0.332062) ($S_3$, 0.271309) ($S_4$, 0.261914) ($S_5$, 0.264352) ($S_6$, 0.260861) ($S_7$, 0.303536) ($L_1$, 0.331835) ($L_2$, 0.328958) ($L_3$, 0.275502) ($L_4$, 0.290804) ($L_5$, 0.291355) ($F_1$, 0.258922) ($F_2$, 0.276608) ($F_3$, 0.324125) ($F_4$, 1.34389) ($F_5$, 0.628596) ($C_1$, 0.399875) ($C_2$, 0.757802) ($C_3$, 0.35326)};

   		 \legend{SHAPE,WARP,VF,HF}
  		\end{axis}
\end{tikzpicture}